\documentclass[acmsmall]{acmart}

\usepackage{graphicx, caption, subcaption}
\usepackage[ruled,vlined]{algorithm2e}
\SetKwComment{Comment}{/* }{ */}

\setcopyright{none}
\renewcommand\footnotetextcopyrightpermission[1]{} 
\makeatletter
\def\runningfoot{\def\@runningfoot{}}
\def\firstfoot{\def\@firstfoot{}}
\makeatother

\AtBeginDocument{%
  }
    
\newtheorem{theorem}{Theorem}[section]
\newtheorem{lemma}[theorem]{Lemma}

\begin{document}

\title{Revisiting the Mapping of Quantum Circuits: Entering the Multi-Core Era}

\author{Pau Escofet}
\authornote{Both authors contributed equally to this research.}
\email{pau.escofet@upc.edu}
\orcid{0000-0003-3372-1931}
\affiliation{%
  \institution{Universitat Politècnica de Catalunya}
  \city{Barcelona}
  \country{Spain}
  \postcode{08034}
}

\author{Anabel Ovide}
\authornotemark[1]
\email{aovigon@upv.es}
\orcid{TODO}
\affiliation{%
  \institution{Universitat Politècnica de València}
  \city{València}
  \country{Spain}
  \postcode{46022}
}

\author{Medina Bandic}
\email{m.bandic@tudelft.nl}
\orcid{0000-0003-4670-0988}
\affiliation{%
  \institution{Delft University of Technology (QuTech)}
  \city{Delft}
  \country{The Netherlands}
  \postcode{2628CJ}
}

\author{Luise Prielinger}
\email{l.p.prielinger@tudelft.nl}
\affiliation{%
  \institution{Delft University of Technology (QuTech)}
  \city{Delft}
  \country{The Netherlands}
  \postcode{2628CJ}
}

\author{Hans van Someren}
\email{j.vansomeren-1@tudelft.nl}
\orcid{0000-0003-4763-6455}
\affiliation{%
  \institution{Delft University of Technology (QuTech)}
  \city{Delft}
  \country{The Netherlands}
  \postcode{2628CJ}
}

\author{Sebastian Feld}
\email{s.feld@tudelft.nl}
\orcid{0000-0003-2782-1469}
\affiliation{%
  \institution{Delft University of Technology (QuTech)}
  \city{Delft}
  \country{The Netherlands}
  \postcode{2628CJ}
}

\author{Eduard Alarcón}
\email{eduard.alarcon@upc.edu}
\orcid{0000-0001-7663-7153}
\affiliation{%
  \institution{Universitat Politècnica de Catalunya}
  \city{Barcelona}
  \country{Spain}
  \postcode{08034}
}

\author{Sergi Abadal}
\email{abadal@ac.upc.edu}
\orcid{0000-0003-0941-0260}
\affiliation{%
  \institution{Universitat Politècnica de Catalunya}
  \city{Barcelona}
  \country{Spain}
  \postcode{08034}
}

\author{Carmen G. Almudéver}
\email{cargara2@disca.upv.es}
\orcid{0000-0002-3800-2357}
\affiliation{%
  \institution{Universitat Politècnica de València}
  \city{València}
  \country{Spain}
  \postcode{46022}
}

\renewcommand{\shortauthors}{Escofet P., Ovide A. et al.}

\begin{abstract}
    Quantum computing represents a paradigm shift in computation, offering the potential to solve complex problems intractable for classical computers. Although current quantum processors already consist of a few hundred of qubits, their scalability remains a significant challenge. Modular quantum computing architectures have emerged as a promising approach to scale up quantum computing systems. This paper delves into the critical aspects of distributed multi-core quantum computing, focusing on quantum circuit mapping, a fundamental task to successfully execute quantum algorithms across cores while minimizing inter-core communications. We derive the theoretical bounds on the number of non-local communications needed for random quantum circuits and introduce the Hungarian Qubit Assignment (HQA) algorithm, a multi-core mapping algorithm designed to optimize qubit assignments to cores with the aim of reducing inter-core communications. Our exhaustive evaluation of HQA against state-of-the-art circuit mapping algorithms for modular architectures reveals a $4.9\times$ and $1.6\times$ improvement in terms of execution time and non-local communications, respectively, compared to the best performing algorithm. HQA emerges as a very promising scalable approach for mapping quantum circuits into multi-core architectures, positioning it as a valuable tool for harnessing the potential of quantum computing at scale.
\end{abstract}



\maketitle

\section{Introduction}
\label{sec:intro}
Quantum computing has emerged as a new computational paradigm, harnessing the unique properties of quantum mechanics, including superposition and entanglement \cite{nielsen_chuang_2010}, to revolutionize problem-solving. These quantum properties enable quantum computers to perform certain calculations at an unprecedented speed, addressing problems previously considered intractable for classical computers. The potential applications of quantum computing span a wide range of domains, from cryptography through algorithms like Shor's prime factorization \cite{shor_polynomial_1997} to optimized database searches using Grover's algorithm \cite{grover_fast_1996} and even the simulation of complex physical systems \cite{Low2019hamiltonian}.

Despite the promise of quantum computing, the current landscape is characterized by a significant gap between the potential of this technology and its practical realization. Quantum computers rely on various qubit implementation technologies, including superconducting qubits \textcolor{black}{\cite{Nakamura_1999, RevModPhys.93.025005}}, photonic qubits \textcolor{black}{\cite{kok_linear_2007, srivastava_2015_optically}}, quantum dots \textcolor{black}{\cite{imagog_quantum_1999, 10.1063/1.5115814}}, and trapped ions \textcolor{black}{\cite{cirac_quantum_1995, PRXQuantum.2.020343}}. However, irrespective of the qubit technology employed, today's quantum computers are limited to a few hundred qubits \cite{chow_2021_ibm}, far from the million-qubit scale required for tackling real-world problems \cite{preskill_quantum_2018}.

Monolithic single-chip architectures face inherent limitations in scalability due to challenges related to the integration of control circuits and wiring for accessing qubits while maintaining low error rates \cite{NAP25196}. Moreover, increasing the number of qubits within a single processor results in a higher rate of undesirable qubit interactions, leading to issues like crosstalk \cite{ding_systematic_2020}. As a result, scaling up monolithic quantum computers to accommodate a higher number of qubits remains a major challenge, necessitating innovative approaches to overcome this bottleneck.

One promising alternative to the single-core quantum computing architecture is the concept of modular or multi-core quantum processors \cite{Bravyi_2022, jnane_multicore_2022, rodrigo_exploring_2020, smith_scaling_2022, laracuente_modeling_2023}. This approach involves interconnecting multiple moderate size chips or quantum cores (QCores) through classical and quantum-coherent links \cite{gold_entanglement_2021}. By adopting a modular architecture, it becomes feasible to tackle the challenge of scalability while maintaining the benefits of quantum coherence.

However, transitioning from monolithic to multi-core quantum devices introduces multiple new difficulties, with communication between cores standing out as a critical issue. Inter-core communications in multi-core quantum processors are significantly more costly than intra-core communications, leading to complex trade-offs and optimization challenges. This paper addresses the intricate problem of mapping quantum circuits onto multi-core quantum processors, explicitly focusing on minimizing the number of non-local communications, a critical factor in maximizing performance. Few multi-core mapping algorithms have been proposed \cite{baker_time-sliced_2020, bandic_mapping_2023}, and their optimality has not been studied in depth. This work proposes a novel mapping technique, comparing it to existing approaches. More precisely, the contributions of this paper can be summarized as:

\begin{itemize}
    \item We perform a non-local communications characterization for Quantum Random Circuits, in which theoretical upper and lower bounds are obtained. Note that this analysis allows, for the first time, for optimality assessment of different multi-core quantum circuit mapping algorithms.
    \item We propose the Hungarian Qubit Algorithm (HQA), originally presented in \cite{Escofet_2023}, and conduct a design exploration improving the algorithm's performance $1.33\times$ on structured circuits.
    \item We compare different state-of-the-art multi-core mapping algorithms, assessing the number of non-local communications and the execution time. We show that HQA outperforms the execution time and non-local communications of the state-of-the-art best performing algorithm by $4.9\times$ and $1.6\times$, respectively.
\end{itemize}

The remainder of this paper is structured as follows. Section \ref{sec:modular_architectures} provides a brief introduction to modular quantum computing architectures and discusses the challenges of these scalable systems that include the need of developing novel quantum circuit mapping techniques. In Section \ref{sec:non-local_comms}, we delve into the task of distributing quantum states from random quantum algorithms into quantum cores, establishing theoretical upper and lower bounds on the number of non-local communications required. Section \ref{sec:mapping_multi-core} conducts a comprehensive review and analysis of state-of-the-art mapping algorithms designed for multi-core quantum computing architectures, comparing their performance to the previously established theoretical bounds. Section \ref{sec:HQA} introduces a novel multi-core mapping algorithm, the HQA, and evaluates its performance against the theoretical bounds. In Section \ref{sec:methodology}, a series of experiments comparing various multi-core mapping algorithms are presented, assessing their scalability in terms of execution time and non-local communications. Finally, in Section \ref{sec:conclusions}, we discuss the results obtained and outline potential paths for future research in this critical domain.

\section{On modular quantum computing architectures}
\label{sec:modular_architectures}
Current quantum computers have successfully integrated \textcolor{black}{up to a thousand} qubits within a single processor, marking a significant milestone in the field \textcolor{black}{\cite{chow_2021_ibm, gambetta_2023_the}}. However, to address real-world problems effectively, quantum computers must scale to operate with thousands or even millions of qubits \cite{preskill_quantum_2018}. This ambitious objective highlights the need for quantum computers to be scaled, containing more qubits.

Nevertheless, scaling quantum computers to accommodate such a large number of qubits is a challenging task. A major difficulty is the intricate integration of classical control circuits and precise wiring for individual qubit addressability, all while maintaining low error rates \cite{NAP25196}. Additionally, the quest to increase qubit counts must be accomplished without increasing crosstalk or interference among qubits \cite{ding_systematic_2020}. Consequently, expanding monolithic quantum computers to integrate an increasing number of qubits remains a critical challenge that requires the exploration of innovative solutions.

One promising architectural alternative involves the division of the Quantum Processing Unit (QPU) into smaller, more manageable cells known as quantum cores or QCores \cite{Bravyi_2022, jnane_multicore_2022, smith_scaling_2022, laracuente_modeling_2023}. Different levels of modularity are envisioned at increasing system complexity, all requiring the introduction of means for communication. More precisely, for superconducting quantum processors, it was first proposed \cite{Bravyi_2022} to interconnect different chips through classical communication links, with the aim of executing large quantum algorithms (i.e. circuits that require more qubits than there are in a single-core) using circuit cutting and knitting techniques \cite{tang_2021_cutqc, piveteau2023circuit}. Next, in the near term, very short quantum links between adjacent chips (i.e. chip-to-chip quantum connector) will be introduced, allowing for quantum communication through two-qubit gates across processors \cite{gold_entanglement_2021}. As this technology evolves, later developments of multi-core quantum computing architectures envision the interconnection of QCores through longer quantum coherent communication links as well as classical channels, enabling the coupling and entanglement of qubits situated in different cores \cite{gold_entanglement_2021,Qintranet} and ultimately across different quantum computers (i.e. between different fridges). Alternatively, some designs focus on generating and distributing entangled qubit pairs to different cores, making use of quantum and classical communication channels \cite{rodrigo_exploring_2020}.

\begin{figure}
\centering
\begin{subfigure}[t]{0.32\textwidth}
    \includegraphics[width=\textwidth]{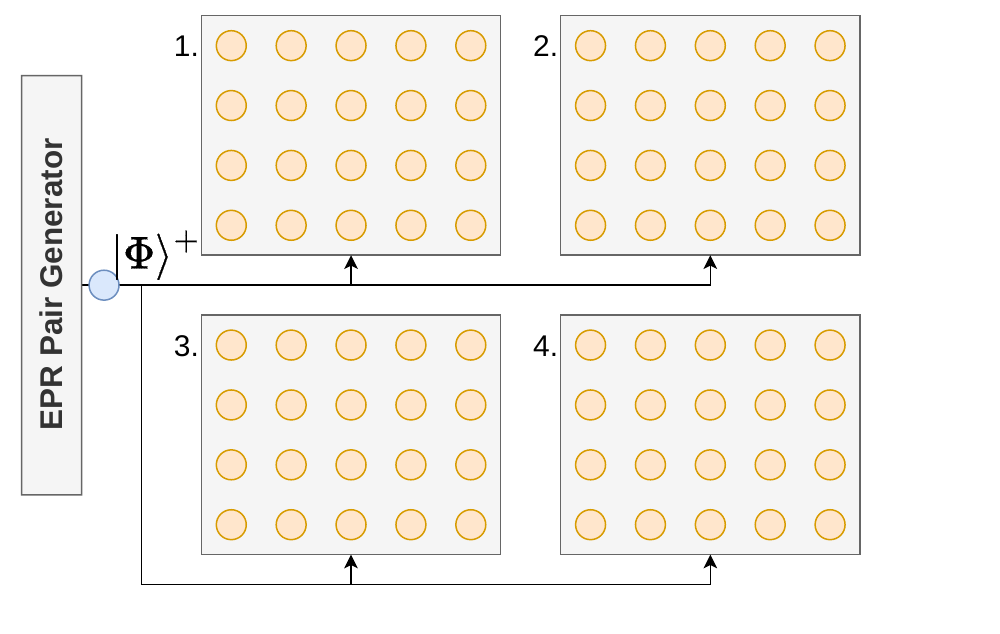}
    \caption{Structure of the architecture, with an entangled state $|\Phi \rangle^+$ at the EPR pair generator.}
    \label{fig:multi_core_1}
\end{subfigure}
\hfill
\begin{subfigure}[t]{0.32\textwidth}
    \includegraphics[width=\textwidth]{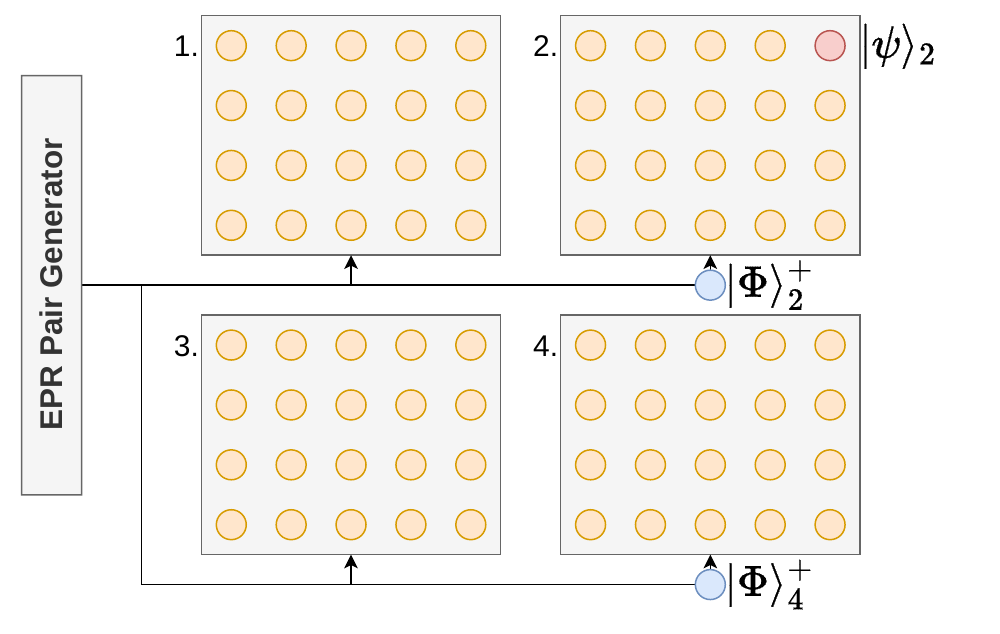}
    \caption{The entangled state is distributed to cores two and four.}
    \label{fig:multi_core_2}
\end{subfigure}
\hfill
\begin{subfigure}[t]{0.32\textwidth}
    \includegraphics[width=\textwidth]{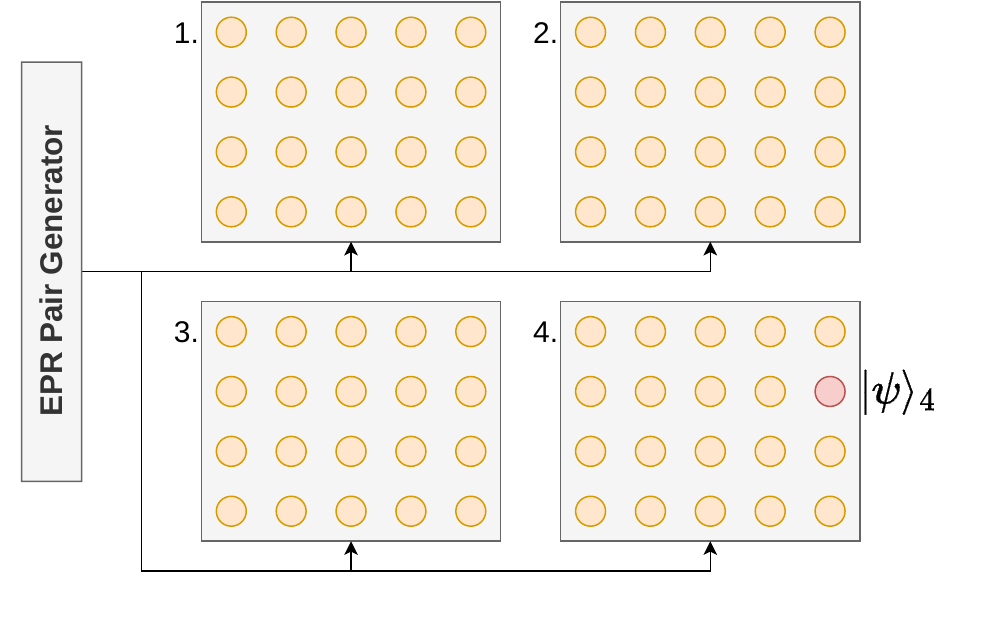}
    \caption{Movement of state $|\psi \rangle$ to core four by using the entangled state.}
    \label{fig:multi_core_3}
\end{subfigure}
\caption{Multi-Core Quantum Computing Architecture based on the distribution of EPR Pairs. \textcolor{black}{Several cores (rectangles) containing qubits (orange circles) are connected to an EPR Pair Generator, in charge of generating and distributing entangled pairs to different cores. Once the entangled states are in the core, they will be used (measured and thus consumed) for communication purposes, enabling the transfer of quantum states across cores.}}
\label{fig:multi_core_architecture}
\end{figure}

In this work, we focus on multi-core quantum computing architectures consisting of several units that can communicate based on the generation and distribution of EPR pairs \textcolor{black}{\cite{PhysRev.47.777}}, as proposed by Rodrigo et al. \cite{rodrigo_exploring_2020}, and illustrated in Figure \ref{fig:multi_core_architecture}. Within this architectural framework, diverse quantum cores are linked to an EPR pair generator through a quantum network, facilitating the distribution of entangled pairs. These entangled states ($|\Phi \rangle^+$) are used to transmit quantum states from one core to another, employing the principles of quantum teleportation \cite{gottesman_1999_demonstrating}, a quantum communication protocol depicted in Figure \ref{fig:teleportation}.

\begin{figure}
\centering
\begin{subfigure}[t]{0.49\textwidth}
    \raisebox{0.35cm}{\includegraphics[width=\textwidth]{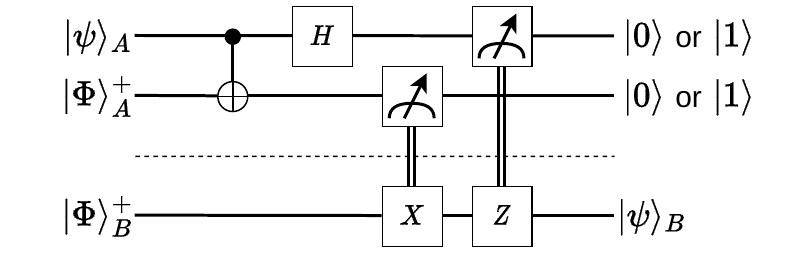}}
    \caption{Teleportation Circuit.}
    \label{fig:teleportation}
\end{subfigure}
\hfill
\begin{subfigure}[t]{0.49\textwidth}
    \includegraphics[width=\textwidth]{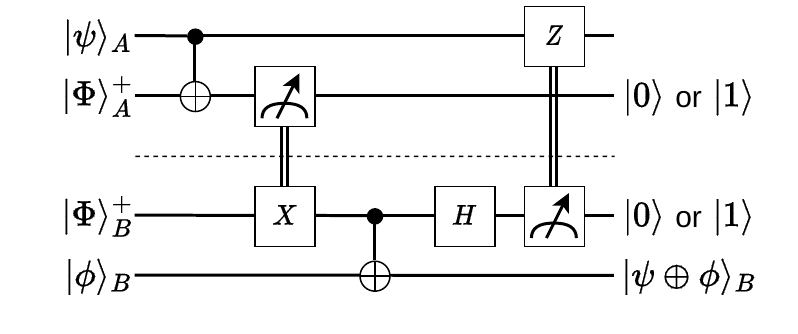}
    \caption{Remote Gate Circuit.}
    \label{fig:remote_gate}
\end{subfigure}
\caption{EPR-based Communication Protocols.}
\label{fig:epr_communications}
\end{figure}

Beyond quantum teleportation, entangled pairs also serve as the primary resource for additional communication protocols, including the execution of remote quantum operations (illustrated in Figure \ref{fig:remote_gate}). These protocols are highly used over long-range quantum networks in Distributed Quantum Computing (DQC) \cite{ferrari_modular_2023}. Note that although remote two-qubit gates can also be used as communications primitives in short-range multi-core architectures, in this work, we focus on quantum circuit mapping algorithms that consider quantum teleportation as a communication means.

\subsection{Challenges for Multi-Core Quantum Computing Architectures}
Multi-core quantum computing architectures represent a promising paradigm for overcoming the limitations associated with scaling monolithic quantum processors to accommodate larger number of qubits. However, this transition has its own difficulties. This section delves into the key problems that must be addressed when designing and implementing multi-core quantum computing systems. These challenges encompass the entire spectrum of quantum computing, from hardware considerations to software and performance evaluation.

\subsubsection{Rethinking the Full-Stack}
Full-stack quantum computing systems have been developed to bridge quantum algorithms with current monolithic quantum processors. However, going to modular architectures will require to redesign such a stack to extend it beyond computation, encompassing also communication. This double full-stack architecture \cite{rodrigo_exploring_2020} necessitates the integration of not only quantum computation elements but also support for classical and quantum communication such as the synchronization and scheduling of quantum/classical information exchange between cores.

\subsubsection{Balancing Computation and Communication Qubits}
In an EPR-based multi-core architecture, qubits must be utilized for both communication and computation, introducing a delicate trade-off. Achieving the right balance between qubits dedicated to computation and those reserved for communication is vital. Overallocation of qubits for communication may limit computational capabilities, while underallocation can restrict the efficient quantum state distribution and manipulation across cores.

\subsubsection{Communication Networks}
The establishment of robust quantum/classical communication networks is central to the success of multi-core quantum computing architectures \cite{Escofet2023InterconnectFF}. This challenge encompasses the development of technologies for the implementation of quantum coherent links capable of transmitting quantum states with minimal decoherence \cite{Kurpiers_2018}. Additionally, the creation of efficient quantum communication primitives and protocols is essential for orchestrating the seamless exchange of quantum information across cores.

\subsubsection{Benchmarking and Performance Metrics}
How to properly measure the performance of a quantum computer is still an open question. In the last years, there has been an effort to define a set of benchmarks and performance metrics for monolithic quantum processors \cite{cross_2019_validating, tomesh_2022_supermarq}. However, these may not capture the intricacies of modular quantum computing architectures as they are missing the communication part as well as the parallelization ability. More precisely, these metrics should include, for instance, factors such as inter-core quantum state transfer latency, communication overhead, and resource utilization efficiency. Accurate evaluation methods are critical for guiding the design and optimization of multi-core quantum computing systems.

\subsubsection{Quantum Compilers for Multi-Core Quantum Computers}
Multi-core quantum computing architectures introduce a complex compilation landscape. Quantum compilers play a pivotal role not only in translating high-level quantum programs into executable instructions but also in performing some modifications to the quantum circuit to deal with the computing hardware constraints, a process known as mapping\textcolor{black}{, where the circuit is transformed to an equivalent one that complies with the restrictions of the targeted quantum processor}. Compilers for multi-core architectures must consider the intricacies of inter-core communication, \textcolor{black}{virtual} qubit movement, and synchronization. Adapting existing compilation techniques to cater to the distributed nature of multi-core architectures is a non-trivial challenge. Developing quantum compilers capable of optimizing quantum circuits across multiple cores while minimizing non-local communications is an active area of research.

In conclusion, transitioning to multi-core quantum computing architectures is required for scaling up quantum computers, but it comes with a set of formidable challenges involving hardware, communication infrastructure, benchmarking, and software development. Addressing them is essential to unlock the full capabilities of these modular quantum systems and pave the way for the next generation of quantum computing technologies.

\subsection{Mapping of Quantum Circuits}
Linked to the last challenge posed in the previous section, mapping is a critical step in the compilation process for quantum circuits. Prior to their execution, quantum circuits are modified \textcolor{black}{(gate decomposition, circuit optimization, and addition of gates such as SWAPs to route the qubits, among other phases)} to be adapted to the hardware's restrictions.


In monolithic quantum computers, the mapper assigns each quantum state (or virtual qubits) from the circuit to an initial physical qubit within the architecture, illustrated in Figure \textcolor{black}{\ref{fig:initial_placement}}. Additionally, it inserts the necessary operations (mostly SWAP gates) to facilitate the movement of quantum states, ensuring that the hardware's specific coupling constraints are met so that each two-qubit gate can be performed, as depicted in Figure \ref{fig:mapped_circuit}. An example of such coupling topology is depicted in Figure \ref{fig:qc_topology}. The scarce connectivity between qubits is the major limitation of current quantum processors. This intricate mapping process is essential for the circuit to function seamlessly on the target quantum processor. Some examples of quantum mappers for monolithic quantum computers are \cite{li_2019_tackling, Qiskit, Amy_2020, lao_2022_timing, Sivarajah_2021}.


\begin{figure}
\vspace{-0.5cm}
\centering
\begin{subfigure}[t]{0.13\textwidth}
    \raisebox{0.45cm}{\includegraphics[width=\textwidth]{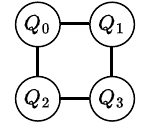}}
    \vspace{-0.8cm}
    \caption{Coupling map of a quantum processor}
    \label{fig:qc_topology}
\end{subfigure}
\hfill
\begin{subfigure}[t]{0.225\textwidth}
     \raisebox{0.15cm}{\includegraphics[width=\textwidth]{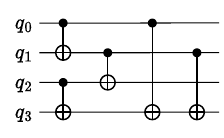}}
    \vspace{-0.8cm}
    \caption{Quantum circuit with 4 (virtual) qubits and 5 two-qubit gates}
    \label{fig:quantum_circuit}
\end{subfigure}
\hfill
\begin{subfigure}[t]{0.275\textwidth}
    \includegraphics[width=\textwidth]{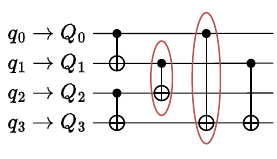}
    \vspace{-0.8cm}
    \caption{Virtual qubits mapped to physical qubits. Unfeasible gates are highlighted}
    \label{fig:initial_placement}
\end{subfigure}
\hfill
\begin{subfigure}[t]{0.33\textwidth}
    \raisebox{0.15cm}{\includegraphics[width=\textwidth]{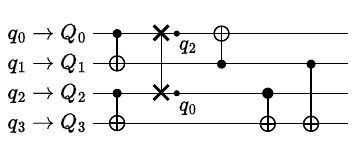}}
    \vspace{-0.8cm}
    \caption{\textcolor{black}{Final mapped circuit. SWAP gates ($\times$) are added to satisfy the connectivity constraints}}
    \label{fig:mapped_circuit}
\end{subfigure}
\vspace{0.3cm}
\caption{Overview of the process of mapping a quantum circuit into the topology of a particular quantum computer.}
\label{fig:monolithic_mapping}
\end{figure}

When going from single-core to multi-core architectures, the mapping problem becomes more
challenging and highly depends on how cores are connected, allowing for some communication
primitives. This work assumes an EPR-distributed architectural model and focuses on distributing quantum states into cores. Along the execution of the circuit, quantum states will be moved from one core to another, ensuring that, every time a two-qubit gate needs to be executed, the involved qubits will be located in the same core.

An example is depicted in Figure \ref{fig:multi-core_mapping}, where a circuit with three timeslices is mapped into a two-core architecture, with two qubits per core. A timeslice is defined as the set of quantum gates from the circuit that can be executed in parallel. In each one of the timeslices, the interacting qubits are located in the same core, ensuring all two-qubit gates will be feasible. The movements across cores will be performed using quantum teleportation (Figure \ref{fig:teleportation}) between timeslices. We refer to these movements across cores as non-local communications\textcolor{black}{, and how these non-local communications are performed depends on how cores are connected among them \cite{Escofet2023InterconnectFF}}.

\begin{figure}[t]
    \includegraphics[width=\textwidth]{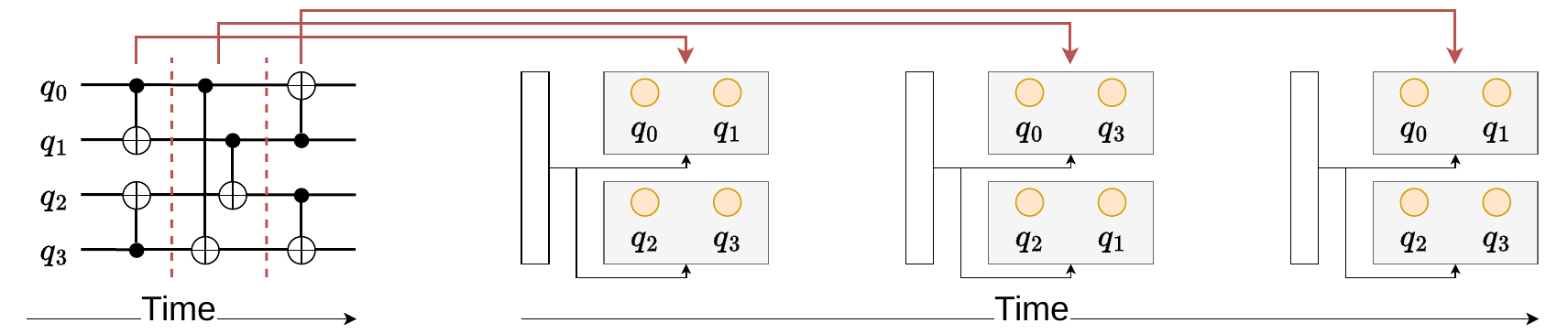}
    \caption{Mapping a 4-qubit quantum circuit (left), into a 2-core Quantum Computing Architecture. For each timeslice (sections of the circuit), we assign qubits into cores, so each two-qubit interaction involves qubits located in the same core.}
    \label{fig:multi-core_mapping}
\end{figure}

\textcolor{black}{In this work’s architectural model, quantum teleportation is used as the inter-core communication protocol, consisting of between four and six quantum gates (depending on the needed corrections). The movement of quantum states across cores will be performed after the mapping, adding the needed quantum gates to perform quantum teleportation. Moreover, a non-local communication requires generating and distributing an entangled pair. In the case where no ancillary qubits are present in the architecture (i.e. the number of virtual qubits in the circuit is the same as the physical qubits in the architecture), at least two qubits devoted to communication per core are required. One will be used as a buffer holding an arriving quantum state, while the other performs the teleportation of a quantum state, freeing space for the other.}

Since non-local communications are much more expensive than intra-core operations, in this work, we focus on the problem of moving qubits across cores, ensuring that, for each two-qubit gate, the involved qubits will be placed in the same core. Such a problem is depicted in Figure \ref{fig:multi-core_mapping} and has been previously studied in \cite{baker_time-sliced_2020} and \cite{bandic_mapping_2023}. Though our architectural model would support Remote Gate execution, for this work, we will only consider qubit movement across cores, making our work analogous to the one in \cite{baker_time-sliced_2020} and \cite{bandic_mapping_2023}.


\section{Non-local Communications Charaterization}
\label{sec:non-local_comms}
Current multi-core mapping algorithms fail to estimate how close the obtained mapping is to the optimal solution. In this section, we focus on Quantum Random Circuits, random algorithms characterized by the number of qubits $q$, the number of gates $g$, and the fraction of two-qubit gates of the circuit $f$. From now on, we will refer to such circuits as $(q, g, f)$-Quantum Random Circuit.

Such circuits are constructed by starting with an empty quantum circuit with $q$ qubits. We then proceed to add $g$ gates to the quantum circuit, each gate has probability $f$ of being a two-qubit gate, and probability $(1-f)$ of being a single-qubit gate. In both cases, the qubits involved in the operation are selected randomly.

\subsection{Non-local Communication Bounds for Random Quantum Circuits}
\label{sec:comms_characterization}

To assess the performance of multi-core mapping algorithms when mapping Quantum Random Circuits, we characterize the number of non-local communications when employing a naive strategy (Theorem \ref{the:naive_upper}), and the lower bound on the number of non-local communications when employing an optimal strategy that does not take into account future qubit interactions (Theorem \ref{the:optimal_bounds}), when mapping a $(q, g, f)$-Quantum Random Circuit into a $(q, N)$-quantum computing architecture (modular architecture with $N$ cores, each of them containing $\frac{q}{N}$ qubits).

\begin{lemma}
For a given $(q, N)$-quantum computing architecture, and a $(q, g, f)$-Quantum Random Circuit, the expected number of qubits involved in unfeasible operations ($q_{unf}$) per timeslice is:

\begin{equation}
    \mathbb{E}(q_{unf})_t = \frac{2(N-1)gfq}{N(q-1)t}
\end{equation}
Where $t$ is the number of timeslices in the circuit, and an unfeasible operation is a two-qubit gate involving qubits that are currently located in different cores.

\label{lemma: expected unfeasible qubits}
\end{lemma}

\begin{proof}
    We begin by assessing the probability that, for a specific two-qubit gate $cx(q_i, q_j)$, the qubits involved, $q_i$ and $q_j$, are located within the same quantum core.

Let $\mathcal{Q}$ be the set of qubits in the whole architecture and $\mathcal{C}(q_i)$ be the set of qubits in the same core as $q_i$. We define the probability of $q_i$ and $q_j$ of being in the same core as the number of qubits different than $q_i$, that are in the same core as $q_i$ \textcolor{black}{(i.e. }$\frac{q}{N}-1$), over the total number of qubits different than $q_i$ \textcolor{black}{(i.e.} $q-1$):

\begin{equation}
    \frac{|\{q_k \in \mathcal{C}(q_i) : q_k \neq q_i\}|}{|\{q_k\in \mathcal{Q} : q_k \neq q_i\}|} = \frac{\frac{q}{N}-1}{q-1}
\end{equation}

As the circuit contains a total of $g$ gates, and the fraction of two-qubit gates is $f$, the expected number of two-qubit gates is given by $g\cdot f$. Let $t$ be the number of timeslices the circuit can be sliced into. As two-qubit gates are randomly distributed across the whole circuit, the expected number of two-qubit gates in each timeslice is:

\begin{equation}
    \mathbb{E}(g_{2q})_t = \frac{g \cdot f}{t}
\end{equation}

By the definition of timeslice, each qubit interacts at most one time in each timeslice, Therefore, the expected number of qubits involved in a two-qubit operation per timeslice is:

\begin{equation}
    \mathbb{E}(q_{inv})_t = 2 \frac{g \cdot f}{t}
\end{equation}

For a given timeslice, the probability of $q_i$ being involved in a two-qubit operation is:

\begin{equation}
    P(q_i \in q_{inv})_t = \frac{2\frac{g \cdot f}{t}}{q}
    \label{eq:qubits_involved_two_qubit_gate}
\end{equation}

Without loss of generality, let us see the possible scenarios of $q_1$ when going from timeslice $t_{i-1}$ to timeslice $t_i$:
\begin{itemize}
    \item $q_1$ interacts in $t_i$ with a qubit $q_k$ \textcolor{black}{in the same core it’s currently in}. In that scenario, no non-local communications are needed, as both interacting qubits are already located in the same core. Here $\mathcal{C}_{t_{i}}(q_1)$ represents the core where $q_1$ is located in timeslice $t_i$.

\begin{equation}
    \mathcal{C}_{t_{i-1}}(q_1) = \mathcal{C}_{t_{i-1}}(q_k) = \mathcal{C}_{t_i}(q_1) = \mathcal{C}_{t_i}(q_k)
\end{equation}

    The probability for the first scenario is given by the probability of $q_1$ interacting in timeslice $t_i$, expressed in Equation (\ref{eq:qubits_involved_two_qubit_gate}), times the number of qubits in $\mathcal{C}_{t_{i-1}}(q_1)$ different than $q_1$, over the number of all qubits different than $q_1$:
\begin{equation}
    P(q_1 \in q_{int} \: \& \: q_1 \notin q_{unf})_t = \frac{2\frac{g \cdot f}{t}}{q} \cdot \frac{\frac{q}{N}-1}{q-1}
\end{equation}
    \item $q_1$ interacts in $t_i$ with a qubit $q_k$ in a different core than it currently is.

\begin{equation}
    \mathcal{C}_{t_{i-1}}(q_1) \neq \mathcal{C}_{t_{i-1}}(q_k) \qquad \qquad \mathcal{C}_{t_i}(q_1) = \mathcal{C}_{t_i}(q_k)
\end{equation}

    The probability for this second scenario is given by the probability of $q_1$ interacting in timeslice $t_i$, expressed in Equation (\ref{eq:qubits_involved_two_qubit_gate}), times the number of qubits not in $\mathcal{C}_{t_{i-1}}(q_1)$, over the number of all qubits different than $q_1$:
    
\begin{equation}
    P(q_1 \in q_{int} \: \& \: q_1 \in q_{unf})_t = \frac{2\frac{g \cdot f}{t}}{q} \cdot \frac{(N-1)\frac{q}{N}}{q-1}
\label{eq:second_scenario}
\end{equation}

    \item $q_1$ does not interact in $t_i$, and therefore no non-local communications are needed. It could happen that $q_1$ is moved to another core to make room for an arriving qubit that needs to interact in the core, but this non-local communication has already been taken into account for the interacting qubit.
\end{itemize}
Therefore, at timeslice $t_i$, the probability of $q_1$ of being involved in an unfeasible two-qubit gate is given by the second scenario, described in Equation(\ref{eq:second_scenario}).

Generalizing to all qubits, we obtain the expected number of qubits involved in unfeasible operations ($q_{unf}$) per timeslice:

\begin{equation}
    \mathbb{E}(q_{unf})_t = q \cdot \frac{2\frac{g \cdot f}{t}}{q} \cdot \frac{(N-1)\frac{q}{N}}{q-1} = \frac{2(N-1)gfq}{N(q-1)t}
\end{equation}

\end{proof}

\begin{theorem}
\label{the:naive_upper}
For a given $(q, N)$-quantum computing architecture and a $(q, g, f)$-Quantum Random Circuit, the number of non-local communications when employing a naive assignation is upper bounded by:

\begin{equation}
    \texttt{non-local comms} \leq \frac{2(N-1)gfq}{N(q-1)}
\end{equation}
\end{theorem}

\begin{proof}
    From Lemma \ref{lemma: expected unfeasible qubits}, we know at each timeslice $\frac{2(N-1)gfq}{N(q-1)t}$ qubits are expected to be involved in unfeasible two-qubit gates. Since each qubit is involved in at most one two-qubit gate, and each two-qubit gate involves two qubits, we have a total of $\frac{(N-1)gfq}{N(q-1)t}$ unfeasible two-qubit gates at each timeslice.

    For each unfeasible gate $cx(q_a, q_b)$, the naive approach will use two non-local communications, one to send $q_a$ to $q_b$'s current core, and one to make space for $q_a$ in the destination core, as cores have a fixed size, and, all the physical qubits the architecture hold a quantum state from the circuit.

    Therefore, from timeslice $t_i$ to timeslice $t_{i+1}$, the expected number of non-local communications caused by the $\frac{(N-1)gfq}{N(q-1)t}$ unfeasible two-qubit gates is upper-bounded by two non-local communications for each unfeasible two-qubit gate:

\begin{equation}
    \texttt{non-local comms}_t \leq 2 \cdot \frac{(N-1)gfq}{N(q-1)t} = \frac{2(N-1)gfq}{N(q-1)t}
\end{equation}

When generalizing for all $t$ timeslices, the number of non-local communications for a $(q, N)$-quantum computing architecture and a $(q, g, f)$-Quantum Random Circuit when using a naive mapping algorithm is upper bounded by:

\begin{equation}
    \texttt{non-local comms} \leq t \cdot \frac{2(N-1)gfq}{N(q-1)t} = \frac{2(N-1)gfq}{N(q-1)}
\end{equation}
\end{proof}

It may happen that when making space for $q_a$ to arrive, the moved qubit (which is chosen randomly) was also a qubit involved in an unfeasible two-qubit gate, and it is moved to the core where the other interacting qubit was. When this happens, two unfeasible two-qubit gates are corrected using just two non-local communications (instead of four, assumed in the previous \textit{proof}), obtaining an average of just one non-local communication per unfeasible two-qubit gate. 

This is why the proposed bound in Theorem \ref{the:naive_upper} is indeed an upper bound and a Naive approach could obtain a lower number of non-local communications. The same reasoning is applied to pose the following Theorem.

\begin{theorem}
\label{the:optimal_bounds}
For a given $(q, N)$-quantum computing architecture, and a $(q, g, f)$-Quantum Random Circuit, the optimal number of non-local communications without using future qubit interactions is lower bounded by:

\begin{equation}
    \frac{(N-1)gfq}{N(q-1)} \leq \texttt{non-local comms}
\end{equation}
\end{theorem}

\begin{proof}
    Similar to Theorem \ref{the:naive_upper}, the lower bound on the number of non-local communications is obtained from the number of unfeasible gates, proposed in Lemma \ref{lemma: expected unfeasible qubits}. However, in this case, no extra communication to make space in the destination core will be needed, assuming the destination core contains a qubit that is also involved in an unfeasible two-qubit gate.

    Therefore, from timeslice $t_i$ to timeslice $t_{i+1}$, the minimum number of non-local communications caused by the $\frac{(N-1)gfq}{N(q-1)t}$ unfeasible two-qubit gates is just one non-local communication for each unfeasible two-qubit gate:

\begin{equation}
    \texttt{non-local comms}_t \geq 1 \cdot \frac{(N-1)gfq}{N(q-1)t} = \frac{(N-1)gfq}{N(q-1)t}
\end{equation}

When generalizing for all $t$ timeslices, the number of non-local communications for a $(q, N)$-quantum computing architecture and a $(q, g, f)$-Quantum Random Circuit without using future qubit interactions is lower bounded by:

\begin{equation}
    \texttt{non-local comms} \geq t \cdot \frac{(N-1)gfq}{N(q-1)t} = \frac{(N-1)gfq}{N(q-1)}
\end{equation}

\end{proof}

This last scenario is encountered when all cores have an even number of qubits involved in an unfeasible two-qubit gate. Within this scenario, only qubits involved in those unfeasible gates will be moved, needing just one teleportation per operation. However, it depends on the mapping algorithm to identify such optimal movements.

The bounds proposed in Theorem \ref{the:naive_upper} and Theorem \ref{the:optimal_bounds} are plotted in Figure \ref{fig:non-local_comms_bounds}, for different circuit sizes, and three different two-qubit gate fractions. Figure \ref{fig:bounds_weak} depicts how the communication bounds vary for a fixed-size architecture with 120 qubits when partitioning it into an increasing number of cores. On the other hand, Figure \ref{fig:bounds_strong} depicts the non-local communications needed when adding 10-qubit cores to the architecture, using a circuit with size depending on the number of qubits of the architecture ($g = 20q$).

\begin{figure}
\centering
\begin{subfigure}[t]{0.48\textwidth}
    \includegraphics[width=\textwidth]{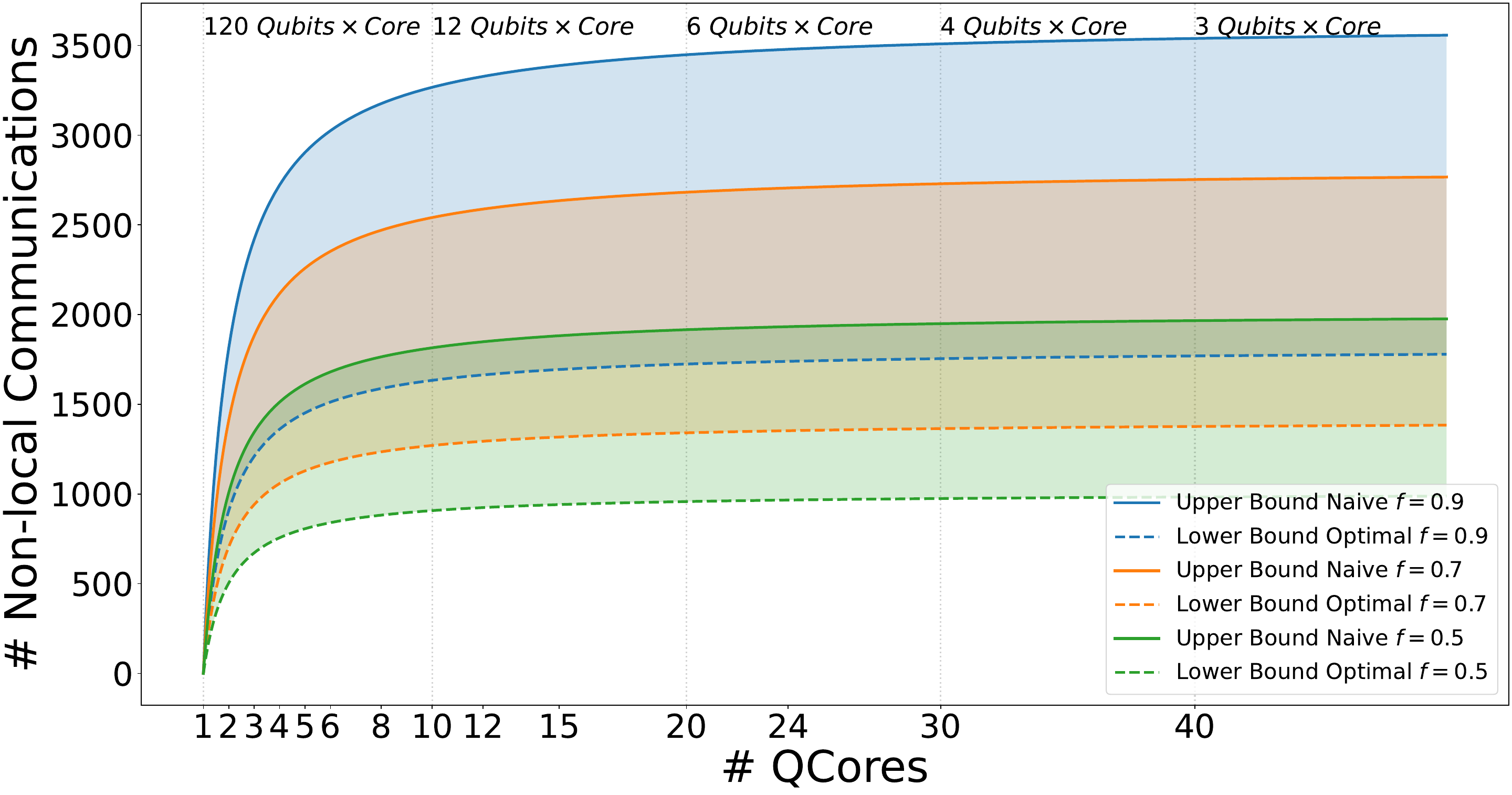}
    \caption{$(120, 2000, f)-$Quantum Random Circuit communication bounds when increasing the number of cores in a fixed-size architecture with 120 qubits.}
    \label{fig:bounds_weak}
\end{subfigure}
\hfill
\begin{subfigure}[t]{0.48\textwidth}
    \includegraphics[width=\textwidth]{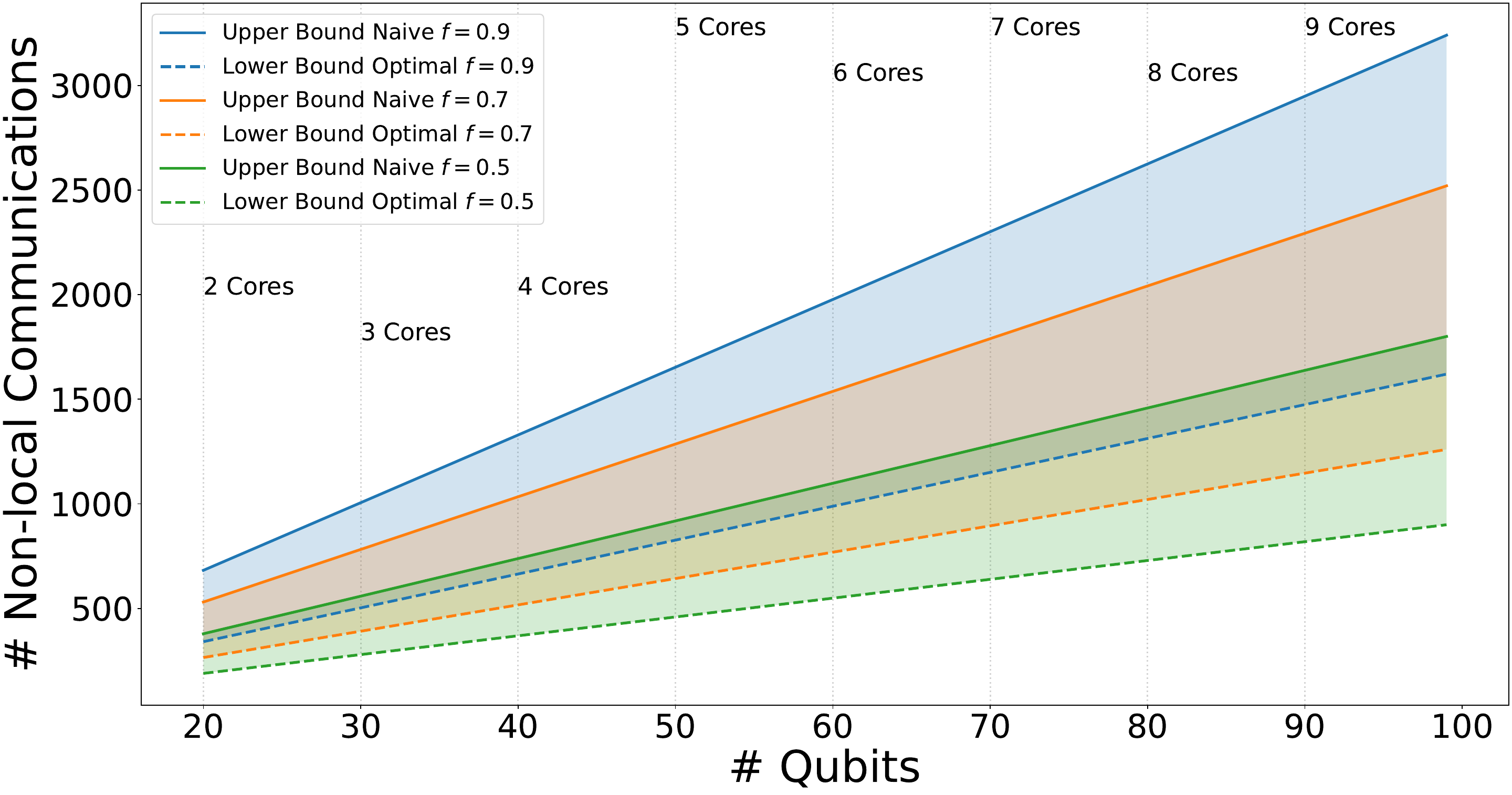}
    \caption{$(q, 20q, f)-$Quantum Random Circuit communication bounds increasing the number of qubits and cores, with a fixed core size of 10 qubits per core.}
    \label{fig:bounds_strong}
\end{subfigure}
\caption{Non-local communications bounds for $(q, g, f)-$Quantum Random Circuits}
\label{fig:non-local_comms_bounds}
\end{figure}

It can be seen how the non-local communications when increasing the number of cores (Figure \ref{fig:bounds_weak}) rapidly increase from going to a monolithic quantum computer (1 core) to a multi-core with a few cores but then stabilizes, showing that, at some point, increasing the number of cores, and thus decreasing the number of qubits per core, has a minimum impact on the number of non-local communications. Figure \ref{fig:bounds_strong} shows that both communication bounds have a linear growth when increasing the architecture's number of cores and the circuit's size.

These bounds provide the first two reliable metrics on the optimal number of non-local communications for quantum random circuits. However, the lower bound proposed in Theorem \ref{the:optimal_bounds} can be further improved by considering future qubit interactions. Therefore, mapping algorithms that take into account qubit interactions to optimize future movements have the potential to achieve a lower number of non-local communications than the one proposed.

\subsection{Naive Mapping Approach}
In order to validate the proposed bounds (Theorem \ref{the:naive_upper} and \ref{the:optimal_bounds}), we propose a Naive mapping algorithm, which, whenever an unfeasible two-qubit gate is encountered, the involved qubits are moved together into one of both involved cores, by randomly making space in the destination core. Therefore, a qubit that was not involved in an unfeasible two-qubit gate will be moved to make space for the arriving qubit. Such an algorithm is described in Algorithm \ref{alg:naive_mapping}.

\begin{algorithm}
\caption{Naive Mapping Algorithm}
\label{alg:naive_mapping}
\KwData{Circuit Timeslice's  $Ts$}
\KwResult{Valid Qubit to Core assignments $As$}
$As[0] \gets $ Random Initial Assignment\;
\For{$T \in Ts$}{ 
    $A_T \gets $ Current Assignment\;
    \For {$(q_A, q_B) \in T$} {
        \If{$A_T[q_A] \neq A_T[q_B]$}{
            $q_{aux} \gets $ Random Qubit from $q_A$'s core\;
            $A_T[q_{aux}] = A_T[q_B]$\;
            $A_T[q_B] = A_T[q_A]$\;
        }
    }
    $A_{T+1} \gets A_T$\;
}
\end{algorithm}

When mapping a $(q,g,f)-$Quantum Random Circuit into a modular architecture using the Naive mapping algorithm proposed above, we expect the number of communications to be lower than the upper bound proposed in Theorem \ref{the:naive_upper}, as some movements to make space for the incoming qubits, will place previously unfeasible qubits, into the right core.

A comparison between the proposed bounds and the Naive mapping approach proposed in Algorithm \ref{alg:naive_mapping} is depicted in Figure \ref{fig:naive_experiments}, where three different Quantum Random Algorithms are mapped to different modular architectures, and the number of non-local communications obtained are compared to the communication bounds proposed in Section \ref{sec:comms_characterization}.

In each experiment, the needed communications are obtained after averaging the communications obtained on twenty different random circuits with the same configuration. The maximum and minimum values obtained in each scenario are shown in the shadowed area.

\begin{figure}
\centering
\begin{subfigure}[t]{0.325\textwidth}
    \includegraphics[width=\textwidth]{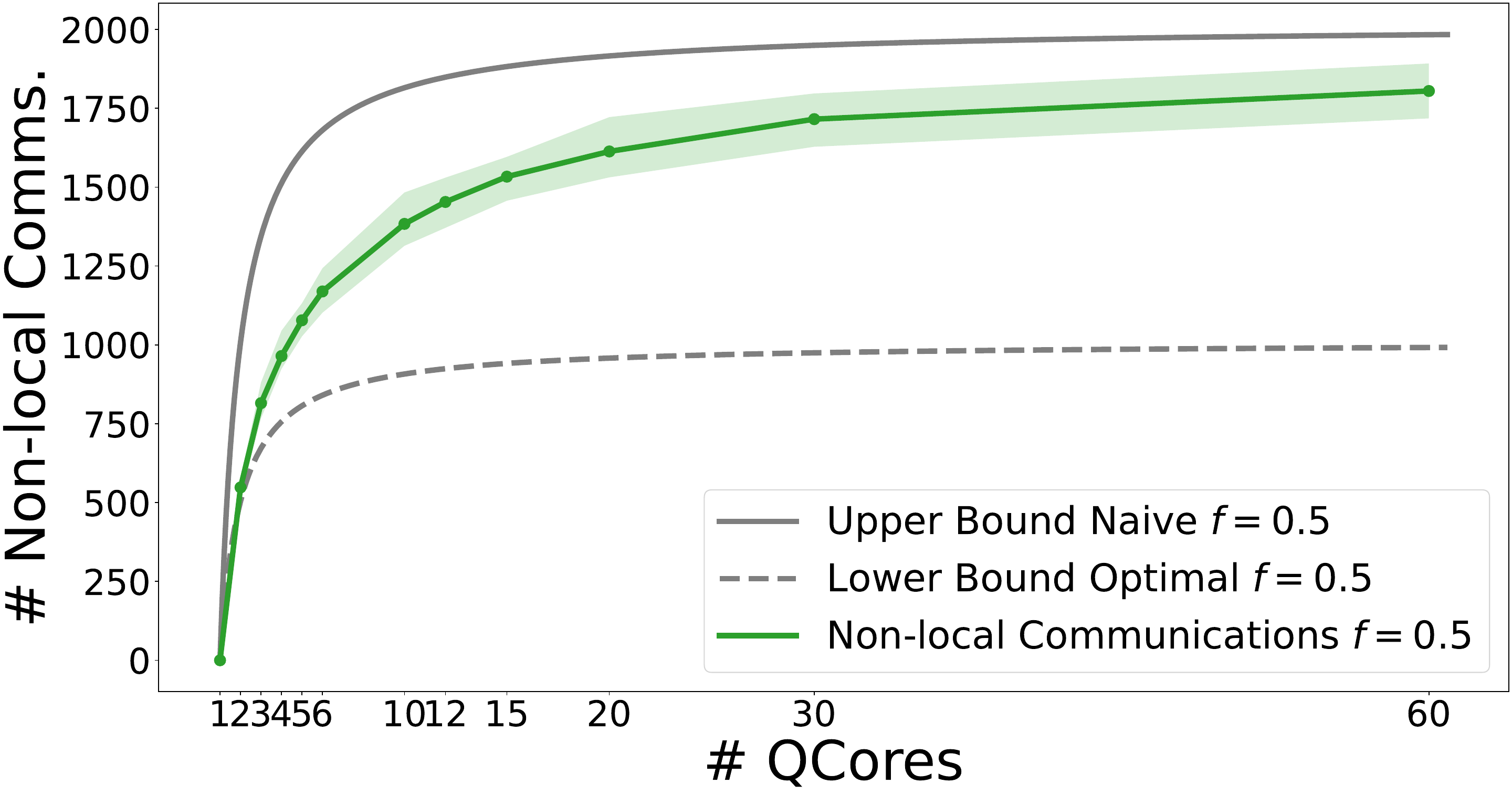}
    \caption{$(120, 2000, 0.5)-$Random Circuit Naive mapping.}
    \label{fig:naive_05}
\end{subfigure}
\hfill
\begin{subfigure}[t]{0.325\textwidth}
    \includegraphics[width=\textwidth]{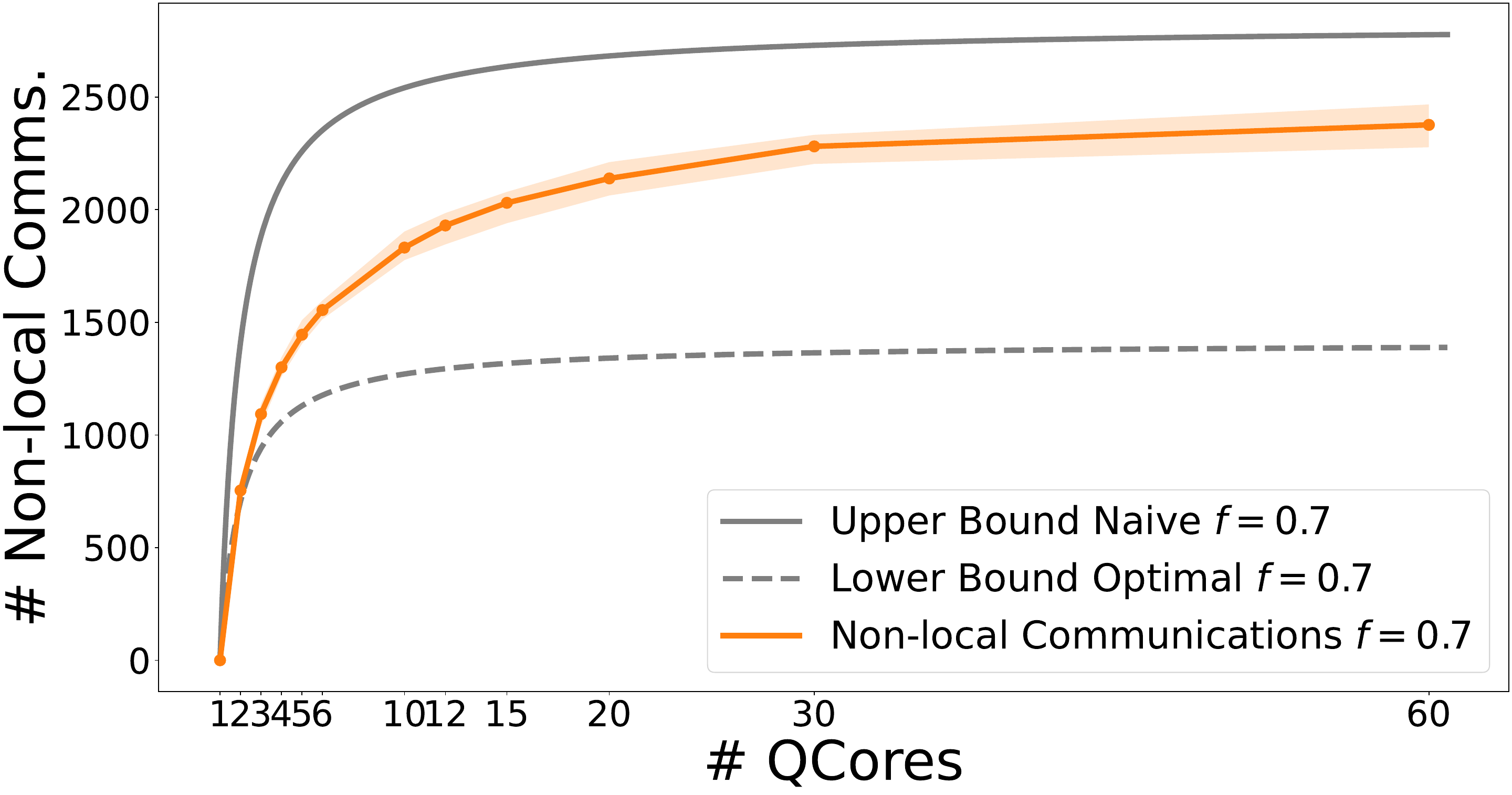}
    \caption{$(120, 2000, 0.7)-$Random Circuit Naive mapping.}
    \label{fig:naive_07}
\end{subfigure}
\hfill
\begin{subfigure}[t]{0.325\textwidth}
    \includegraphics[width=\textwidth]{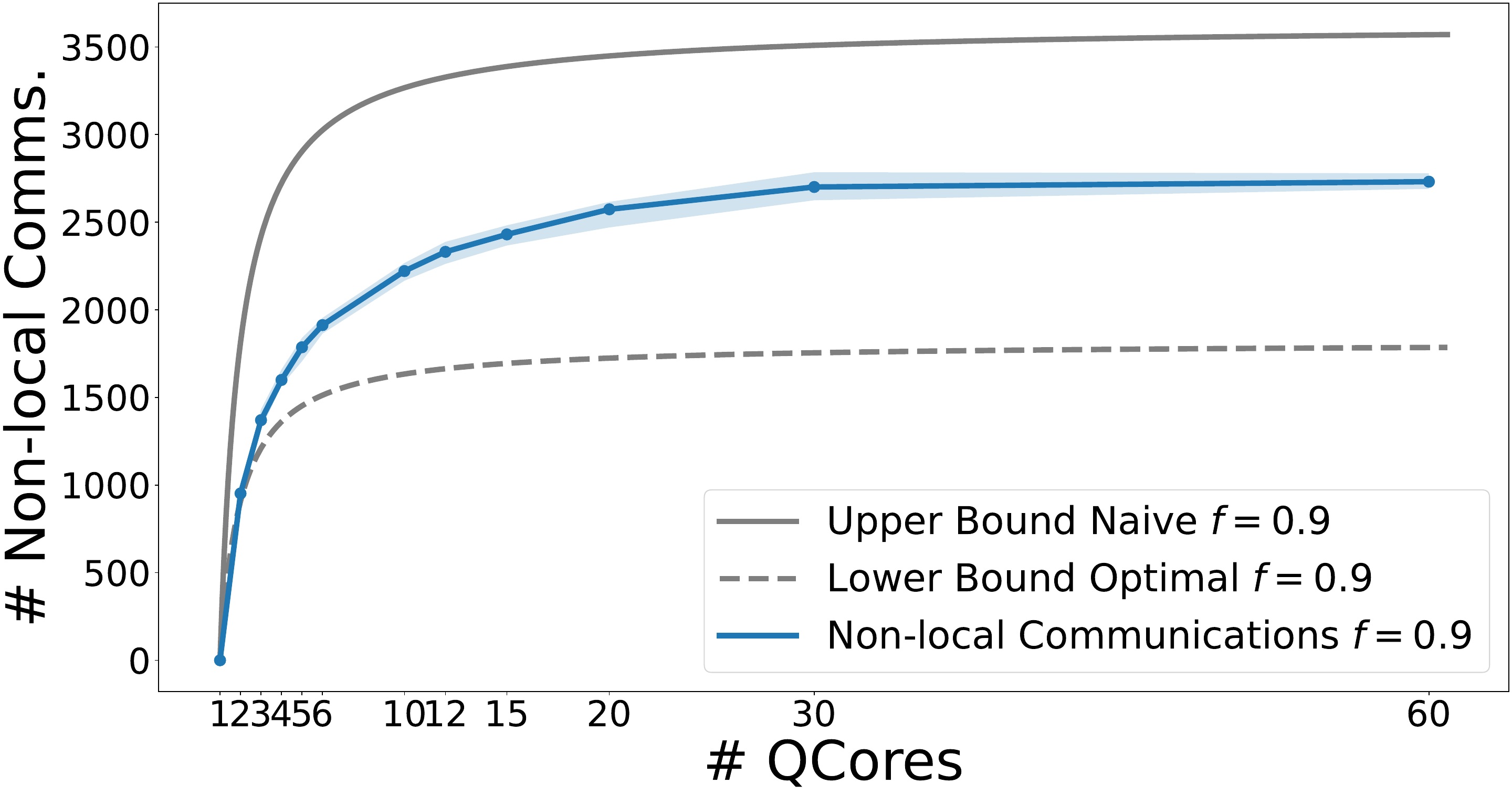}
    \caption{$(120, 2000, 0.9)-$Random Circuit Naive mapping.}
    \label{fig:naive_09}
\end{subfigure}
\caption{Non-local communications when using the Naive Mapping Algorithm for different $(120, 2000, f)-$Quantum Random Circuits. Communication bounds are shown in grey.}
\label{fig:naive_experiments}
\end{figure}

Figure \ref{fig:naive_experiments} shows that the number of non-local communications needed is closer to the upper bound when the two-qubit gate frequency is lower. This is due to the movement of random qubits to make space for the arriving ones. If the selected qubit was involved in an unfeasible two-qubit gate in that same timeslice, we would have moved it anyway, and therefore, the movement to make space has no impact on the number of non-local communications. As the two-qubit gate fraction increases, it is more likely that the selected random qubit is involved in an unfeasible two-qubit gate, achieving a number of non-local communications further from the upper bound (Figure \ref{fig:naive_09}) than when the two-qubit gate fraction is small (Figure \ref{fig:naive_05}).

In future sections, we will use this Naive method as a baseline for the non-local communications needed to map quantum algorithms into multi-core architectures.

\section{Multi-Core Mapping Algorithms Analysis}
\label{sec:mapping_multi-core}
This section reviews the inter-core mapping algorithms proposed in \cite{baker_time-sliced_2020}, called Fine Grained Partitioning (FGP-rOEE), and the algorithm proposed in \cite{bandic_mapping_2023}, based on Quadratic unconstrained binary optimization (QUBO) problem. These two algorithms are some of the scarce algorithms proposed so far to solve the mapping problem for multi-core quantum computers.

\textcolor{black}{Other quantum circuit mapping algorithms for distributed quantum computing have been proposed \cite{ferrari_modular_2023, wu_autocomm_2022, andresmartinez2023distributing}. They make use of  EPR pairs for only remote operations across cores (i.e. telegate) \cite{10.1145/3579367} or also for qubit teleportation (i.e. teledata). In addition, other mapping algorithms such as \cite{zhang2023compilation} are developed to target chiplet architectures \cite{smith_scaling_2022}, a different type of modular architectures than the ones considered in this work. In this work, we only focus on qubit distribution in modular multi-core architectures, in which processors are connected with classical and quantum links, not taking into account remote operations or a different type of processor architecture.}


\subsection{FGP-rOEE}
Baker et al. \cite{baker_time-sliced_2020} proposed the FGP-rOEE algorithm to tackle the mapping problem in multi-core quantum computers. The algorithm takes as inputs a quantum circuit with $q$ qubits and an architecture with $N$ cores, each core accommodating $\frac{q}{N}$ qubits. The algorithm operates under two fundamental assumptions:  i) The coupling map of each core is all-to-all; i.e., all qubits in a core have direct connections to one another, allowing the direct execution of a two-qubit gate if both qubits are located within the same core. ii) Cores are interconnected all-to-all, enabling the exchange of quantum states between any pair of cores.

The primary objective of the algorithm is to obtain a sequence of qubit-to-core assignments, with one assignment per timeslice. The quantum circuit is separated into these timeslices, and a valid assignment of qubits to cores is found for each one. A valid assignment must have every pair of interacting qubits in the timeslice assigned to the same core, ensuring that no two-qubit gate involves qubits located in different cores.

For a given timeslice $t$, FGP-rOEE computes the interaction graph of that timeslice, a graph with the qubits as nodes and an edge between two nodes if the qubits interact with each other in future timeslices. The edges are weighted, representing the immediacy of the interaction. These weights are called look-ahead weights and are computed using Equation (\ref{look-ahead weights}), where $I(m, q_i,q_j) = 1$ if qubits $q_i$ and $q_j$ interact at timeslice $m$, and the exponential decay function ($2{-x}$) is used so nearby timeslices have more impact on the look-ahead weights than latter ones.

\begin{equation}
    w_t(q_i, q_j) = \sum_{t < m \leq T} I(m, q_i,q_j) \cdot 2^{-(m-t)}
    \label{look-ahead weights}
\end{equation}

For qubits that interact exactly at timeslice $t$, a weight of infinity is set to the edges. This weight implies that, at that particular timeslice, these qubits must unequivocally reside within the same core, thus making their separation impossible.

Next, a $k$-partitioning algorithm is employed to partition the interaction graph into $k$ disjoint subsets of nodes. Here, $k$ corresponds to the number of cores $N$, and it is imperative that all partitions have the exact same size, as cores have a fixed size, and we can only assign $\frac{q}{N}$ qubits into each core. Due to these strict constraints, the set of suitable $k$-partitioning algorithms becomes notably limited.

In \cite{baker_time-sliced_2020}, the use of the Overall Extreme Exchange (OEE) algorithm \cite{PARK1995899} is proposed. The OEE algorithm builds upon the Kernighan–Lin \cite{kernighan_1970_efficient} algorithm and expands its capabilities. In their work, Baker et al. introduce a variant of the OEE algorithm known as rOEE, or relaxed Overall Extreme Exchange. Similar to the OEE algorithm, the rOEE also starts with an assignment and performs exchanges of nodes until a valid partition is reached (i.e. all interacting qubits are in the same partition).

This procedure is repeated for every pair of timeslices, using the previous assignment of qubits to cores as input for the rOEE algorithm. With this, a path of valid assignments is found over the whole circuit. A detailed analysis of the FGP-rOEE algorithm can be found in \cite{ovide2023mapping}.

\subsection{QUBO}
Bandic et al. \cite{bandic_mapping_2023} proposed a mapped solution problem for multi-core or modular architectures based on the Quadratic Unconstrained Binary Optimization (QUBO) \cite{original_QUBO} method, relying on prior subgraph isomorphism approaches \cite{subgraph_isomorphism}, and single-core solutions \cite{dury2020qubo}. This approach addresses qubit allocation and inter-core communication costs through binary decision variables, being suitable for different modular architectures. QUBO introduces a mathematical problem classified as NP-hard, which is subject to optimization. The formula employed represents the objective function to be minimized as follows:
\begin{equation}
\text{min}_xx^TQx=\text{min}_x\sum_{i<j}Q_{ij}x_ix_j+\sum_iQ_{ii}x_i
\end{equation}
where the variable $x$ represents a binary decision vector of dimension $N$, and $Q$ designates a symmetric square matrix composed of $N\times N$ real-valued constants.

Similar to FGP-rOEE, QUBO partitions the quantum circuit into slices, each encompassing a series of gates. Each time slice can be depicted as a graph, with nodes representing distinct qubits and edges denoting interactions performed by two-qubit gates. The primary aim of the objective function is to determine an allocation for each time-slice graph, ensuring that all qubits engaged in a common gate are assigned to the same computational core without surpassing the core's capacity. Additionally, it seeks to reduce inter-core communications during these assignments. The objective function is then generalized for all time slices and modified to minimize potential inter-core communication between every pair of assignments. For more information on the mathematical process, refer to  \cite{bandic_mapping_2023}. The objective function counts with a weighting factor denoted as $\lambda$, which can be employed to adjust the different components of the objective function.

\subsection{Evaluation and Limitations}
In this section, we compare the performance of the state-of-the-art mapping algorithms for multi-core quantum computing architectures \cite{baker_time-sliced_2020, bandic_mapping_2023} described in previous sections to the non-local communication bounds proposed in Theorem \ref{the:naive_upper} and \ref{the:optimal_bounds}, as well as to the naive approach proposed in Algorithm \ref{alg:naive_mapping}. 

We use Qiskit's \cite{Qiskit} Random Circuit library to match the available implementation of \cite{bandic_mapping_2023}. The performance of the FGP-rOEE algorithm \cite{baker_time-sliced_2020} is assessed using three different random circuits of 120 qubits and different sizes: Random S (100 timeslices), Random M (200 timeslices), and Random L (300 timeslices). While the performance of the QUBO Mapping algorithm \cite{bandic_mapping_2023} is assessed with smaller random circuits of 48 qubits and different random sizes, Random S (25 timeslices), Random M (50 timeslices), and Random L (75 timeslices). The size difference of the used circuits reflects the execution time each algorithm needs to compute the valid mapping.

Figure \ref{fig:fgp_qubo_bound_experiments} shows the non-local communications when mapping the selected circuit into a modular architecture with a fixed number of qubits (120) distributed over an increasing number of cores. It can be seen how, though both the FGP-rOEE and QUBO mapping algorithms consider future interaction among qubits, taking into account more information than the Naive approach, their performance for Quantum Random Circuits is far from optimal, achieving in some cases a higher number of non-local communications than the Naive approach, and surpassing the Naive upper bound derived in Theorem \ref{the:naive_upper}.

\begin{figure}
\centering
\begin{subfigure}[t]{0.32\textwidth}
    \includegraphics[width=\textwidth]{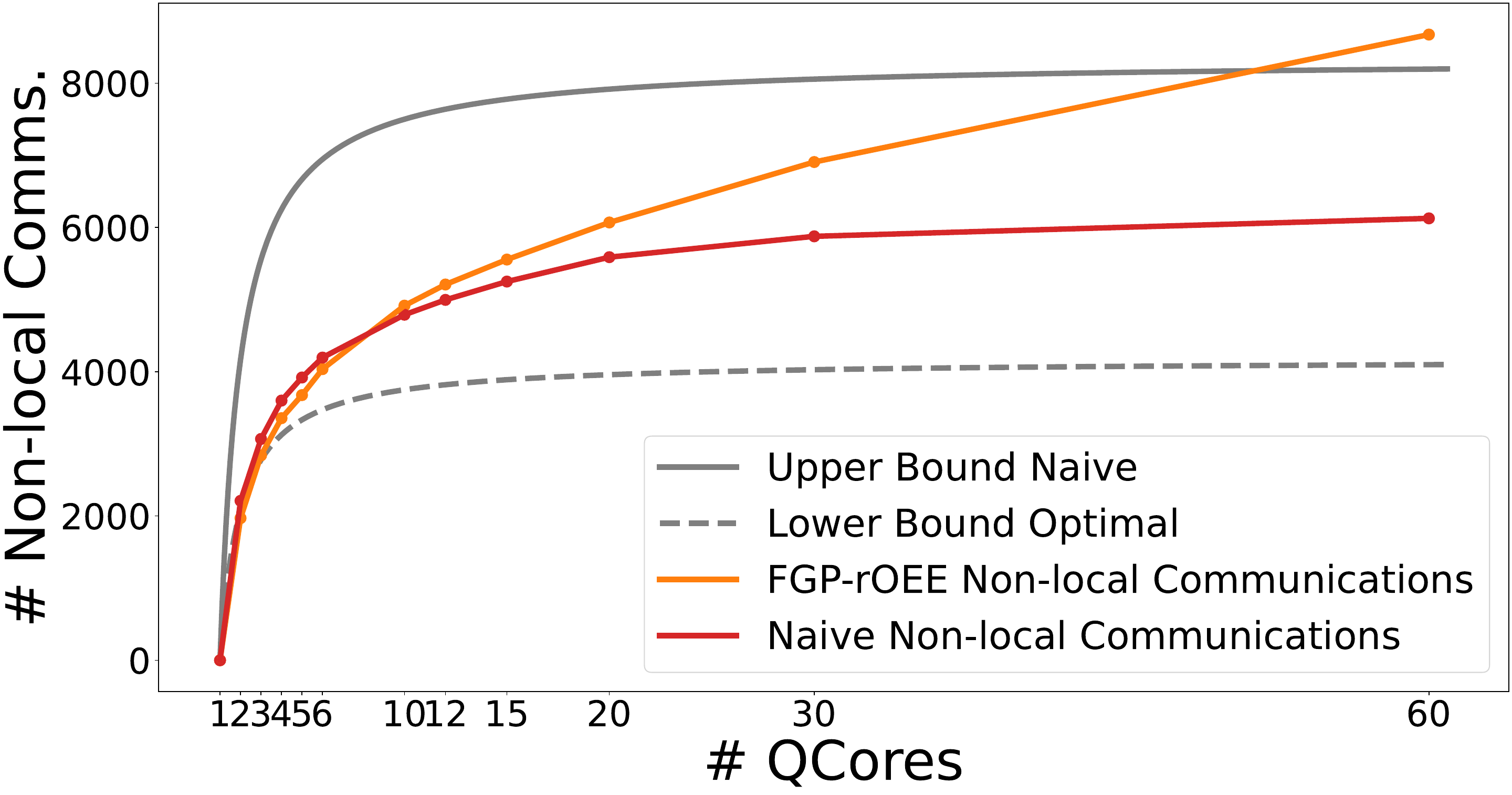}
    \caption{Random S FGP-rOEE Mapping.}
    \label{fig:random_s_fgp-roee}
\end{subfigure}
\hfill
\begin{subfigure}[t]{0.32\textwidth}
    \includegraphics[width=\textwidth]{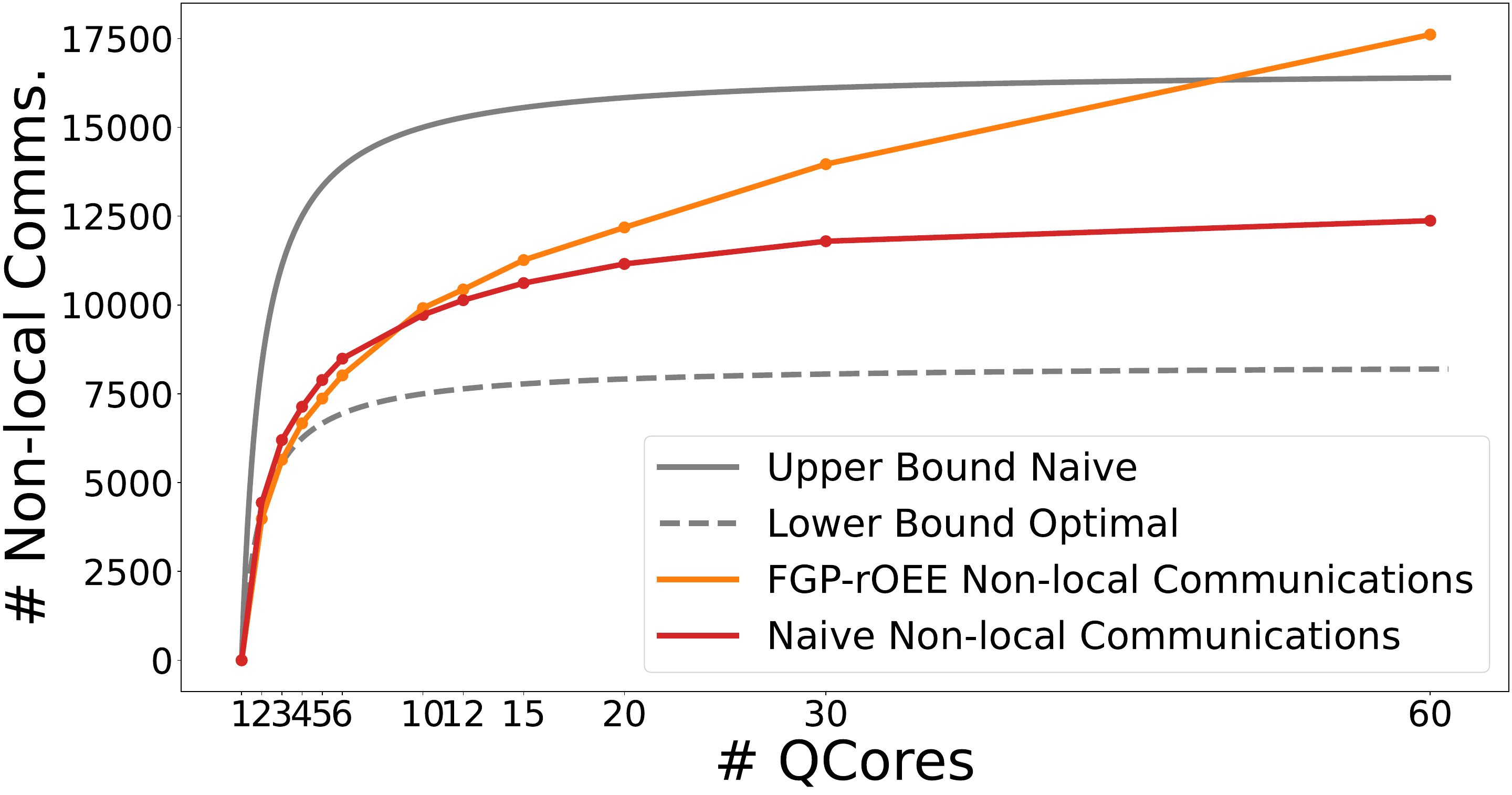}
    \caption{Random M FGP-rOEE Mapping.}
    \label{fig:random_m_fgp-roee}
\end{subfigure}
\hfill
\begin{subfigure}[t]{0.32\textwidth}
    \includegraphics[width=\textwidth]{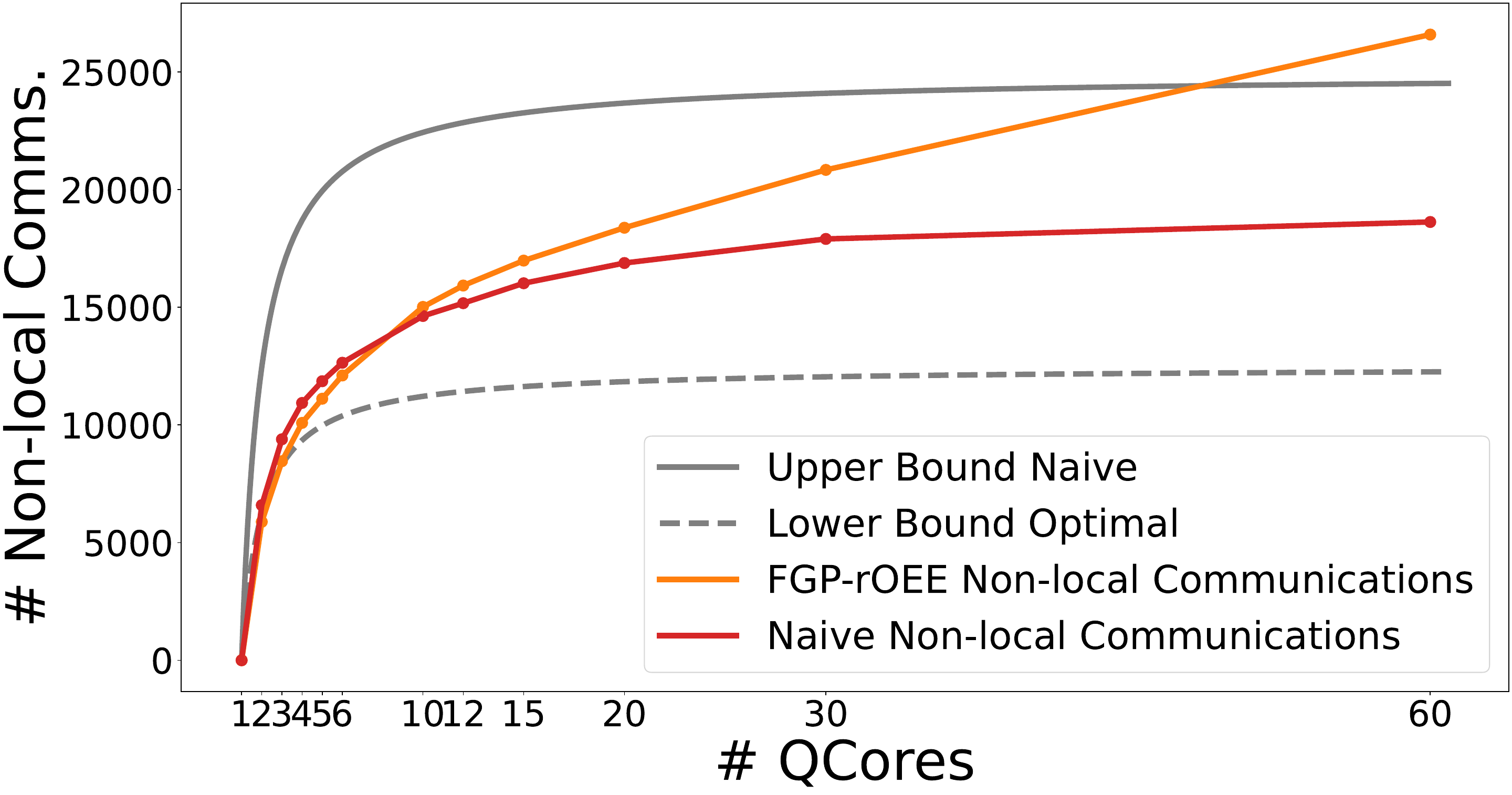}
    \caption{Random L FGP-rOEE Mapping.}
    \label{fig:random_l_fgp-roee}
\end{subfigure}
\hfill
\begin{subfigure}[t]{0.32\textwidth}
    \includegraphics[width=\textwidth]{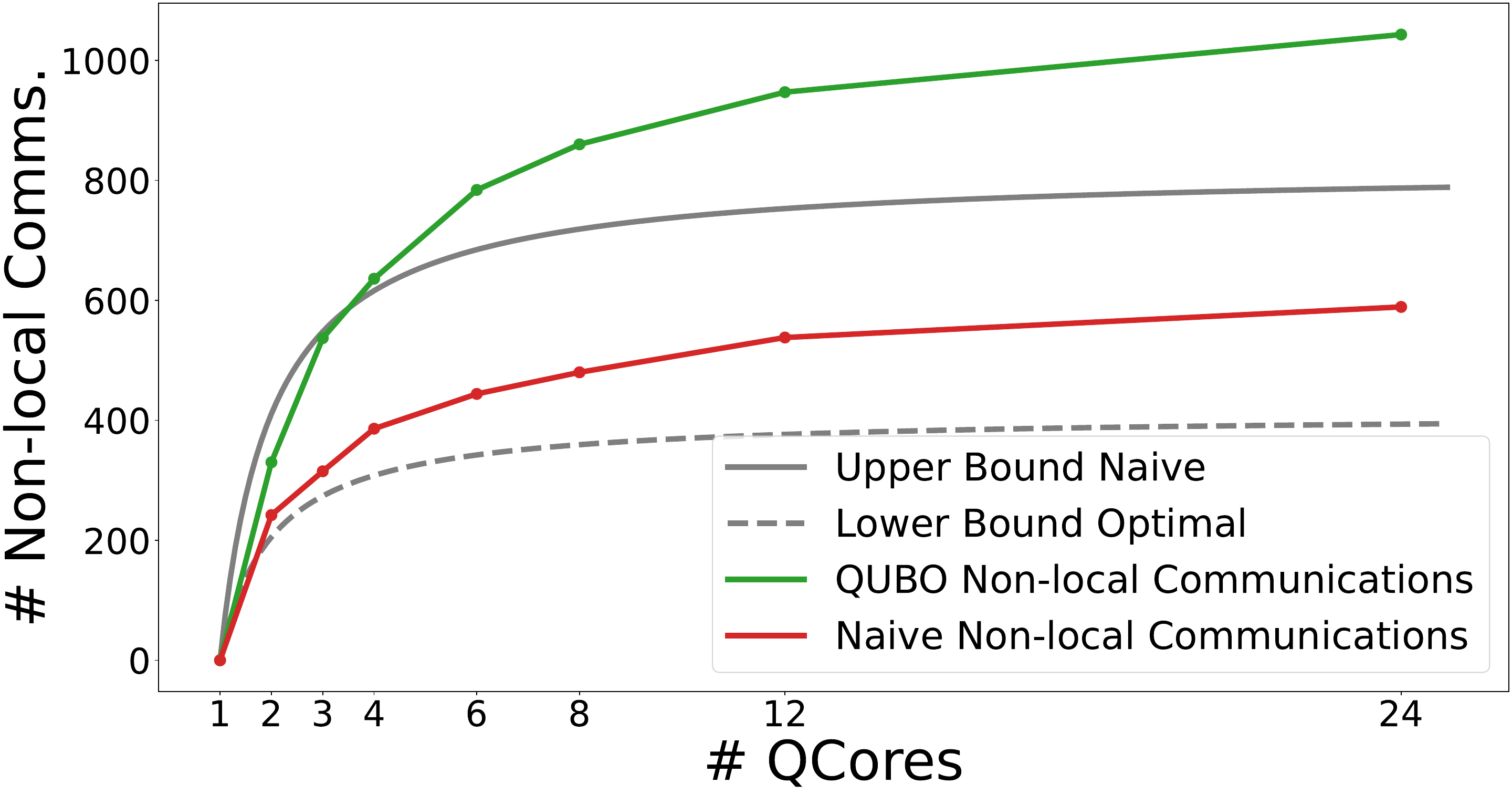}
    \caption{Random S QUBO Mapping.}
    \label{fig:random_s_qubo}
\end{subfigure}
\hfill
\begin{subfigure}[t]{0.32\textwidth}
    \includegraphics[width=\textwidth]{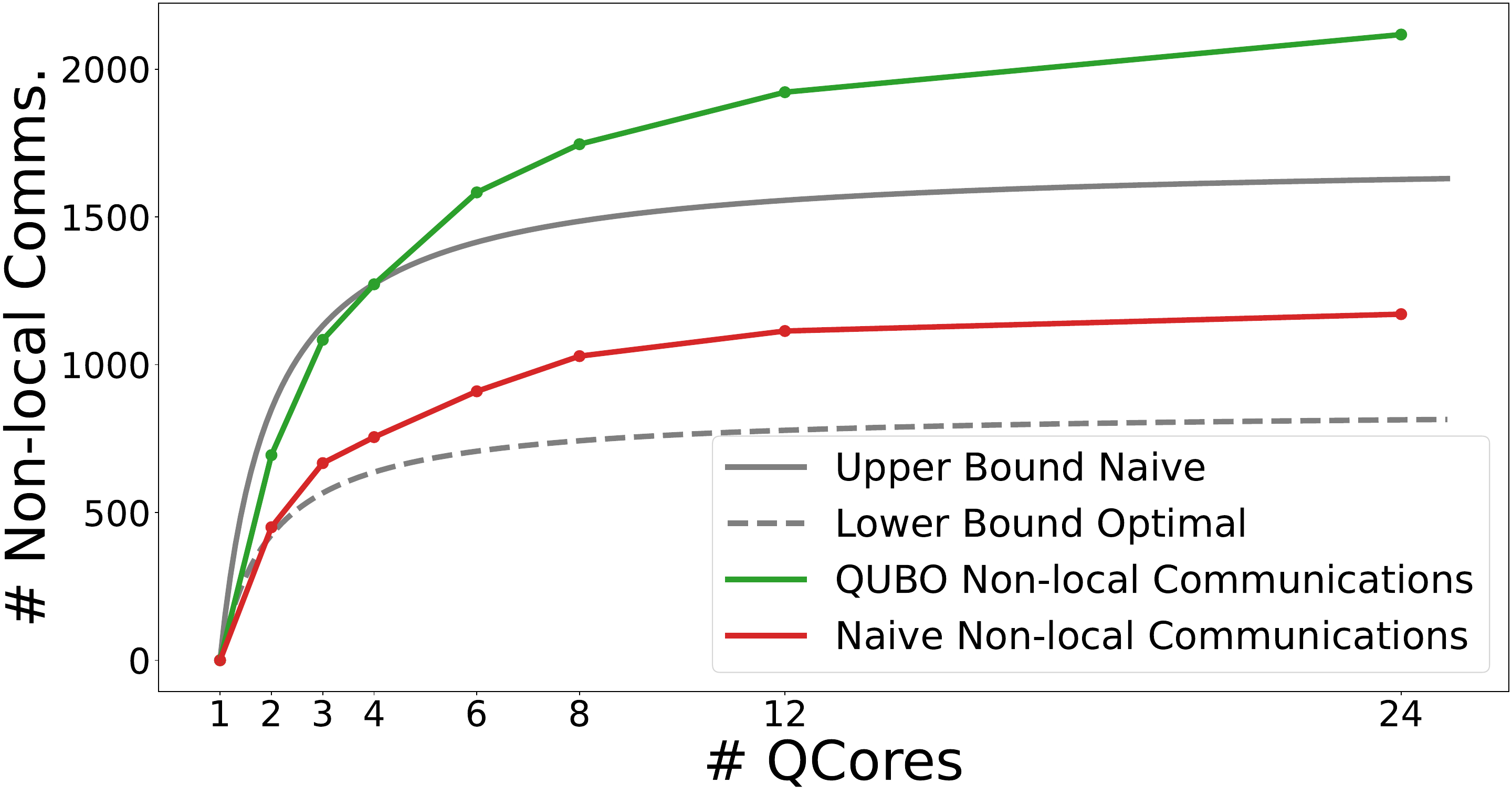}
    \caption{Random M QUBO Mapping.}
    \label{fig:random_m_qubo}
\end{subfigure}
\hfill
\begin{subfigure}[t]{0.32\textwidth}
    \includegraphics[width=\textwidth]{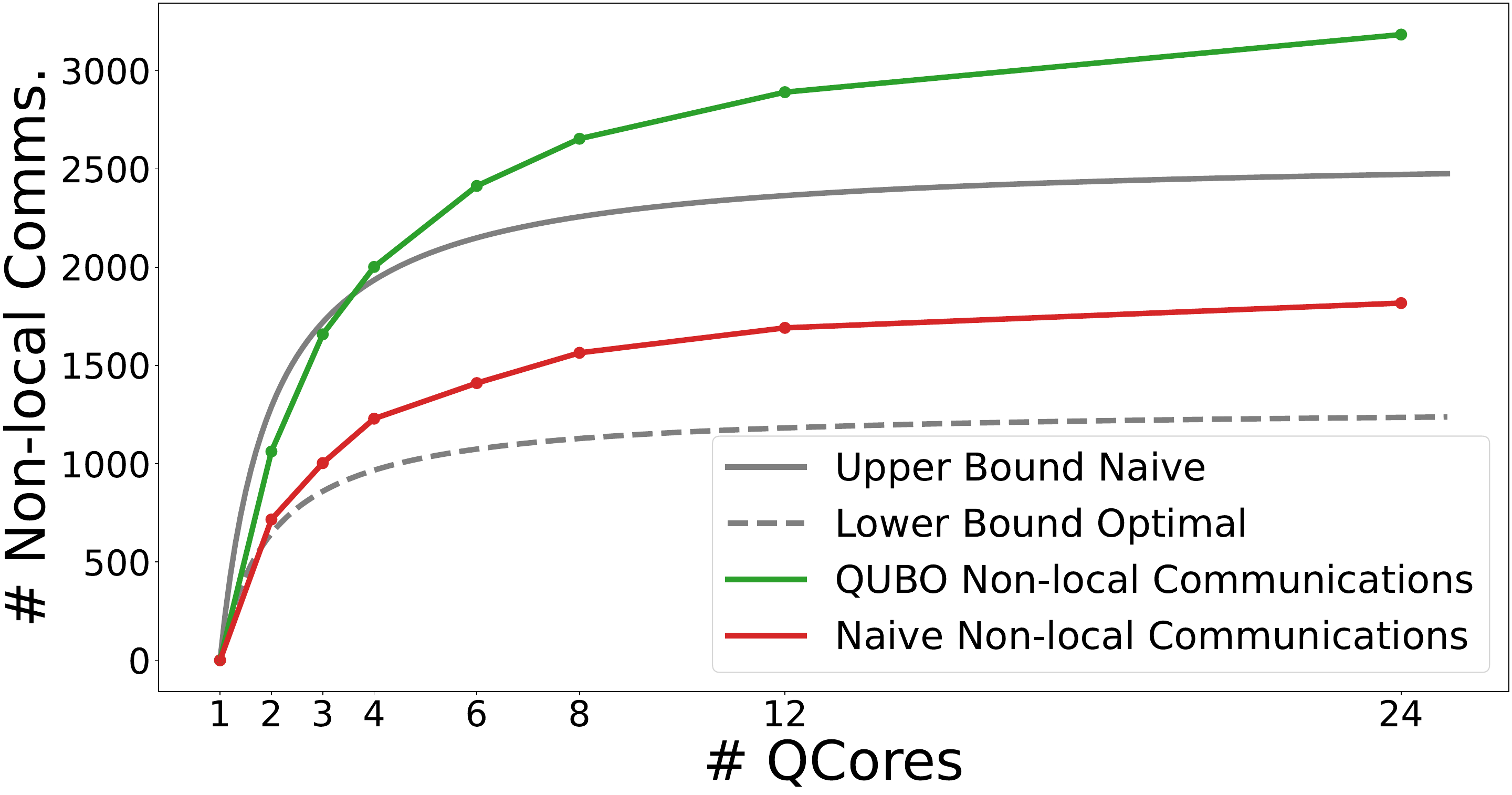}
    \caption{Random L QUBO Mapping.}
    \label{fig:random_l_qubo}
\end{subfigure}
\caption{Non-local communications when using the FGP-rOEE and QUBO Algorithms for different Quantum Random Circuits.}
\label{fig:fgp_qubo_bound_experiments}
\end{figure}

Our intuition on the observed sub-optimal performance of the tested algorithms is that it can be attributed to their fundamental approach, which treats mapping as a graph partitioning problem. This approach, while suitable for certain scenarios, may not be ideal for the specific task of distributing quantum circuits into modular architectures. Quantum circuit mapping, as a problem, indeed shares similarities with graph partitioning, as it involves the allocation of computational resources. However, it introduces unique constraints, such as a fixed number of cores and a predefined number of qubits per core, and existing algorithms for graph partitioning often do not consider the specific constraints imposed by quantum computing architectures. 

Furthermore, the primary objective when mapping a quantum circuit is not merely achieving a balanced allocation of qubits to cores but, more critically, minimizing the number of non-local communications needed in between time slices.

It is worth mentioning that both approaches can potentially improve the lower-bound on the number of non-local communications proposed in Theorem \ref{the:optimal_bounds}, as they consider future qubit interaction and, as stated in Section \ref{sec:comms_characterization}, mapping algorithms that take into account qubit interactions to optimize future movements have the potential to achieve a lower number of non-local communications than the one proposed.

In summary, while graph partitioning techniques offer a foundational framework for tackling the challenge of distributing quantum circuits into modular architectures, their suitability is limited by the unique constraints of mapping. The scarcity of specialized tools and the failure to prioritize the reduction of non-local communications contribute to the underperformance of these approaches. Addressing these limitations is essential to unlocking the full potential of quantum computing systems.

\section{Hungarian Qubit Assignment}
\label{sec:HQA}
In this section, the Hungarian Qubit Assignment (HQA) \cite{Escofet_2023} is described. Similar to \cite{baker_time-sliced_2020} and \cite{bandic_mapping_2023}, the HQA algorithm describes how to assign qubits to cores in between timeslices and is generalized to map the whole algorithm. However, unlike the previously discussed multi-core mapping algorithms, HQA distributes two-qubit operations into cores, a much easier task than distributing qubits into cores by graph partitioning. As each two-qubit operation within a timeslice involves two distinct qubits, the algorithm's output will be a valid assignment for that particular timeslice.


\subsection{General Overview}
The proposed algorithm is summarized in Algorithm \ref{alg:hqa_mapping}, and a simple example for a particular timeslice is depicted in Figure \ref{fig:hqa_overview}.

\begin{algorithm}
\caption{Hungarian Qubit Assignment}
\label{alg:hqa_mapping}
\KwData{Circuit Timeslice's  $Ts$}
\KwResult{Valid Qubit to Core assignments $A$}
$A_0 \gets $ Initial Assignment\;
$G_{unf} \gets []$\; 
\For{$T \in Ts$}{ 
    $A_T \gets $ Current Assignment\;
    \For {$(q_A, q_B) \in T$} {
        \If{$A_T[q_A] \neq A_T[q_B]$}{
            $G_{unf}.\texttt{insert}((q_A, q_B))$\; 
        }
    }

    \While {$G_{unf}.\texttt{not\_empty}()$} {
        $C_{mat} \gets $ Empty Cost Matrix\;
        \For {$(q_A, q_B) \in G_{unf}$} {
            \For{$n \in$ Cores} {
                \uIf{$n$ is full}{
                    $C_{mat}[(q_A, q_B)][n] \gets \infty$\; 
                }
                \uElseIf{$A_T[q_A] == n$ or $A_T[q_B] == n$}{
                    $C_{mat}[(q_A, q_B)][n] \gets 1$\; 
                }
                \Else{
                    $C_{mat}[(q_A, q_B)][n] \gets 2$\; 
                }
            }
        }
        $\texttt{assign} \gets \texttt{Hungarian}(C_{mat})$\;
        $A_T \gets A_T.\texttt{update(assign)}$\;
        $G_{unf} \gets G_{unf}.\texttt{remove(assign)}$\;
    }
    
    $A_{T+1} \gets A_T$\;
}
\end{algorithm}

\begin{figure}
\centering
\begin{subfigure}[t]{0.22\textwidth}
    \raisebox{0.43cm}{\includegraphics[width=\textwidth]{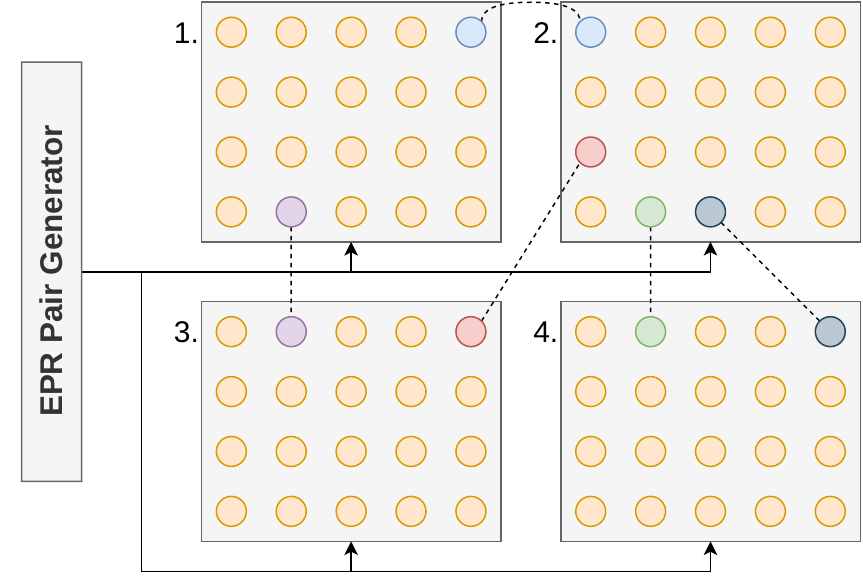}}
    \caption{Unfeasible assignment of qubits to cores.}
    \label{fig:hqa_1}
\end{subfigure}
\hfill
\begin{subfigure}[t]{0.22\textwidth}
    \includegraphics[width=\textwidth]{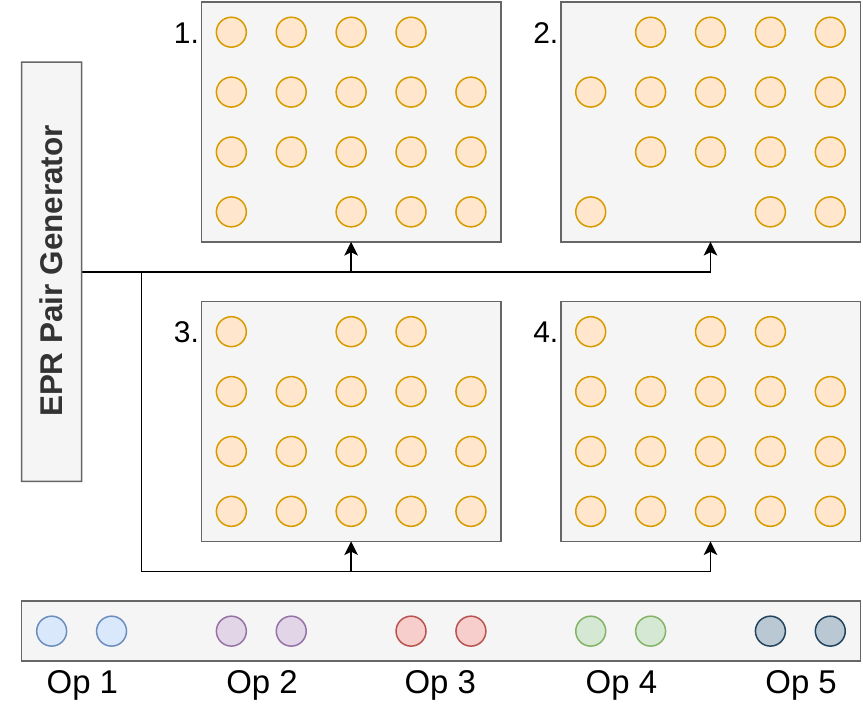}
    \caption{Remove the qubits involved in unfeasible operations.}
    \label{fig:hqa_2}
\end{subfigure}
\hfill
\begin{subfigure}[t]{0.22\textwidth}
    \includegraphics[width=\textwidth]{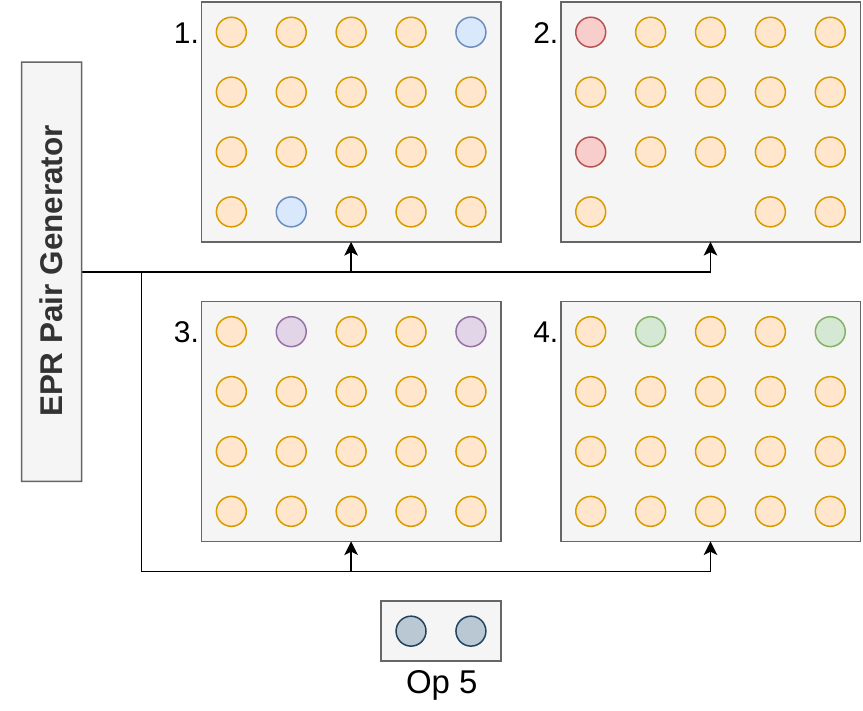}
    \caption{Each core gets assigned one unfeasible two-qubit gate.}
    \label{fig:hqa_3}
\end{subfigure}
\hfill
\begin{subfigure}[t]{0.22\textwidth}
    \raisebox{0.43cm}{\includegraphics[width=\textwidth]{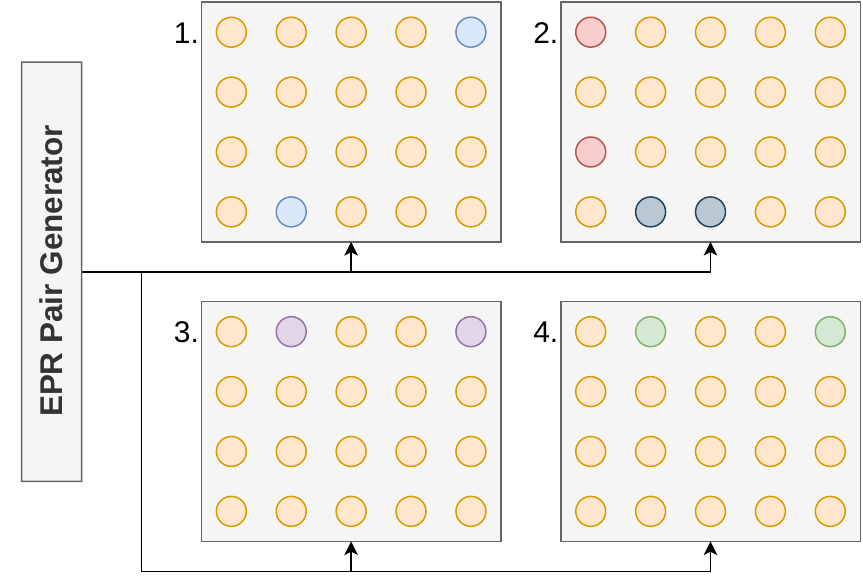}}
    \caption{Repeat until all two-qubit gates are assigned to a core.}
    \label{fig:hqa_4}
\end{subfigure}
\caption{Hungarian Qubit Assignment Overview. }
\label{fig:hqa_overview}
\end{figure}

The algorithm starts with a valid assignment of qubits to cores for timeslice $t$. The following timeslice ($t+1$) has a set of unfeasible two-qubit gates involving qubits that are currently located in different cores. This is depicted in Figure \ref{fig:hqa_1}, where we can see five unfeasible two-qubit operations (color-coded) for an example architecture of four cores.

The qubits involved in the unfeasible two-qubit operations are then removed from the assignment and placed in an auxiliary vector of unassigned qubits, as depicted in Figure \ref{fig:hqa_2}. Now, the task is to assign each unfeasible operation to a core with enough space to take in the qubits. For this assignment, we will construct a cost matrix using the cost function in Equation (\ref{cost funcion}). This cost function assigns a cost value $\mathcal{C}_t$ to each pair of unfeasible two-qubit operation $op_i$ and core $c_j$ based on how many non-local communications are needed to place both qubits involved in the operation $q_A, q_B$ into the destination core $c_j$ (i.e. one if a qubit involved in the unfeasible operation $op_i$ is already placed in core $c_j$, two otherwise).

\begin{equation}
    \mathcal{C}_t (op_i, c_j) = 
\begin{cases}
    \infty & \text{if } c_j \text{ is full}\\
    1 & \text{if } q_A \in c_j \text{ or } q_B \in c_j\\
    2 & \text{otherwise}
\end{cases}
\label{cost funcion}
\end{equation}

We use the Hungarian algorithm \cite{kuhn_hungarian_1955}, a highly efficient linear assignment algorithm with polynomial time complexity ($\mathcal{O} (n^3)$), to assign a single operation to each core using the cost matrix previously constructed. Its versatility, simple implementation, and robustness make the Hungarian algorithm a favored choice in various fields, especially when a quadratic algorithm is impractical.

The objective of the assignment problem is to find the best way to assign a set of tasks (cores) to a group of resources (unfeasible two-qubit operations) while minimizing the total cost. By doing this, at each iteration, only one operation will be assigned to each core, ensuring the core's capacity is not exceeded.

When an operation is assigned to a core, both qubits involved are placed in the free spaces of the core, decreasing by two the number of free spaces in the core every time an operation is assigned to it. By placing both qubits of each unfeasible operation into the same core ensures all two-qubit gates will be feasible in the following timeslice.

The Hungarian algorithm only assigns one operation per core. Therefore, some operations will remain unassigned in case of having more unfeasible operations than cores, as depicted in Figure \ref{fig:hqa_3}. A new cost matrix will be computed for those unassigned operations, considering the new free spaces of each core and setting a weight of infinity for those cores already full, ensuring that no core exceeds its capacity and that the resulting assignment will be valid. This process is repeated until all unfeasible gates have been assigned to a core, resulting in a valid assignment for the timeslice, as shown in Figure \ref{fig:hqa_4}.

It is important to note that when using the same number of virtual qubits (quantum states in the circuit) as physical qubits (qubits in the quantum computing architecture), each core must contain an even number of free spaces for this approach to work. Otherwise, when assigning operations into cores, there will be a pair of qubits left to assign and two cores with exactly one free space each, making it impossible to assign both qubits of the unfeasible operation to the same core. To ensure all cores contain an even number of free spaces and that all operations will be assigned, for each pair of cores with an odd number of space, an auxiliary two-qubit gate involving two non-interacting qubits from those cores is created. Forcing all cores to have an even number of free spaces.

Regarding the initial assignment for the algorithm, a structured assignment has the potential to perform much better than a random assignment. This possibility is explored in future sections.

\subsection{Considering future qubit interactions}
Future interactions of qubits can be added to the cost matrix to perform an assignment that further reduces the number of non-local communications. To this end, we quantify how much qubits interact in future timeslices using the same approach as in \cite{baker_time-sliced_2020}, described in Equation (\ref{look-ahead weights}), and introduce the attraction force of a qubit $q_i$ to a core $c_j$, which is computed as

\begin{equation}
    \texttt{attr}^q_t (q_i, c_j) = \sum_{i'=0}^q J_t(q_{i'}, c_j) \cdot w_t(q_i, q_{i'})
\label{attraction weights}
\end{equation}

where $J_t(q_{i'}, c_j) = 1$ if qubit $q_{i'}$ is in core $c_j$ at timeslice $t$, and $w_t(q_i, q_{i'})$ is computed using Equation (\ref{look-ahead weights}).

The new cost matrix is then computed using the costs given in Equation (\ref{cost function with attraction}), where the number of non-local communications needed and the attraction forces are combined. Note that, as each operation $op_i$ involves two qubits ($q_A$ and $q_B$), the operation's attraction force to a core is the average attraction force of the involved qubits.

\begin{equation}
    \texttt{attr}^{op}_t(op_i, c_j) = \frac{\texttt{attr}^q_t(q_A, c_j) + \texttt{attr}^q_t(q_B, c_j)}{2}
\end{equation}

\begin{equation}
    \mathcal{C}_t (op_i, c_j) = 
\begin{cases}
    \infty & \text{if } c_j \text{ is full}\\
    1 - \texttt{attr}^{op}_t(op_i, c_j) & \text{if } q_A \in c_j \text{ or } q_B \in c_j\\
    2 - \texttt{attr}^{op}_t(op_i, c_j) & \text{otherwise}
\end{cases}
\label{cost function with attraction}
\end{equation}

The proposed approach cost function is completely tunable, allowing for more complex variations of the problem. For example, if not all cores are connected to each other, the number of non-local communications to move a qubit to a given core may be more than one. The cost function can be adapted to this case and many others, leading to a robust and widely applicable inter-core mapping algorithm.

\subsection{Discussion}
\label{sec:HQA_discussion}
Similar to the other two mapping algorithms discussed, we compare the HQA algorithm to the bounds on the non-local communications proposed in Theorem \ref{the:naive_upper} and \ref{the:optimal_bounds}.

Figure \ref{fig:hqa_bound_experiments} shows that, unlike the other mapping algorithms, HQA always performs better than the Naive approach. Staying always below the upper bound, and obtaining a lower number of communications than the lower bound for a small number of cores.

\begin{figure}
\centering
\begin{subfigure}[t]{0.32\textwidth}
    \includegraphics[width=\textwidth]{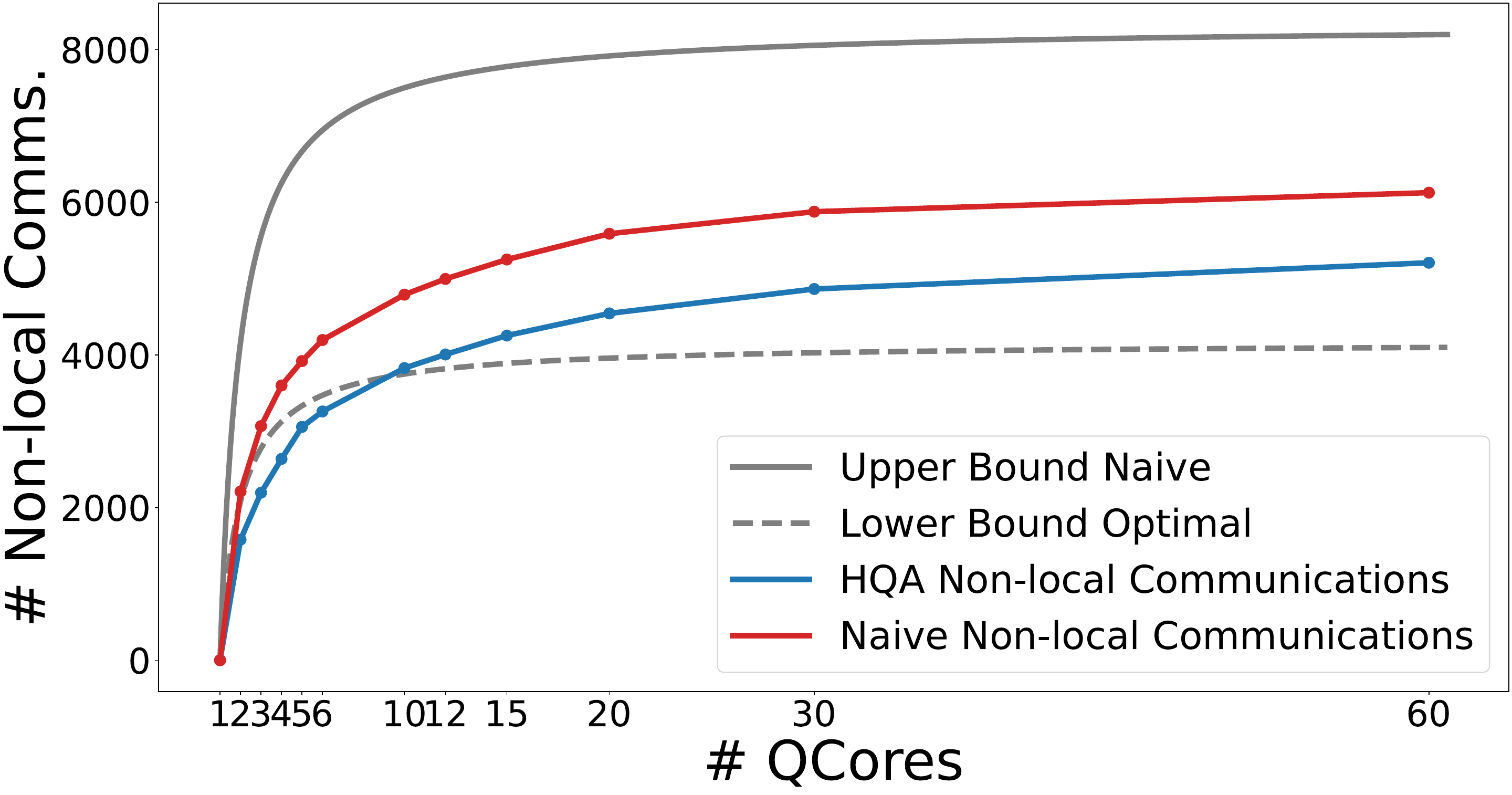}
    \caption{Random S HQA Mapping.}
    \label{fig:random_s_hqa}
\end{subfigure}
\hfill
\begin{subfigure}[t]{0.32\textwidth}
    \includegraphics[width=\textwidth]{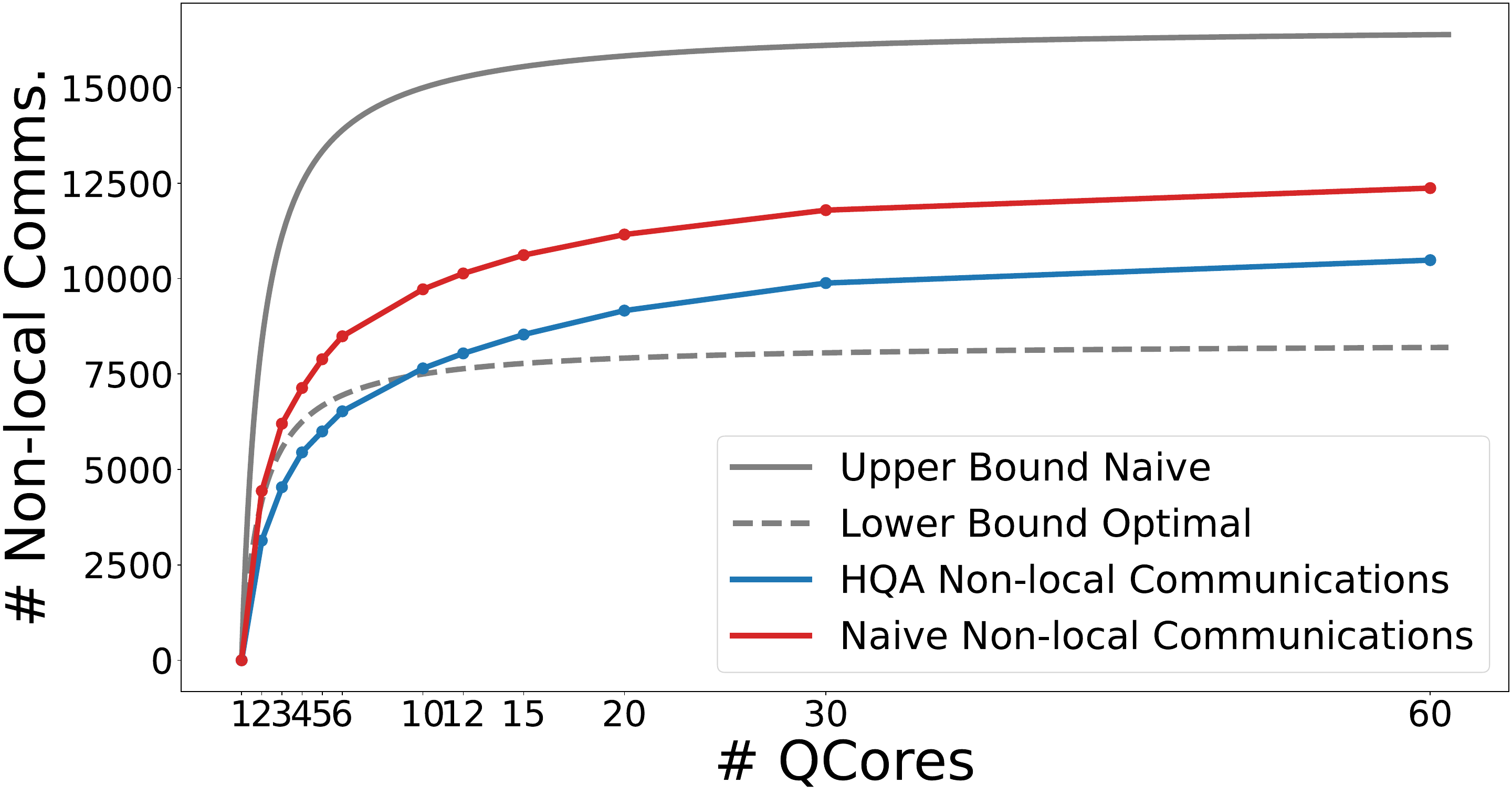}
    \caption{Random M HQA Mapping.}
    \label{fig:random_m_hqa}
\end{subfigure}
\hfill
\begin{subfigure}[t]{0.32\textwidth}
    \includegraphics[width=\textwidth]{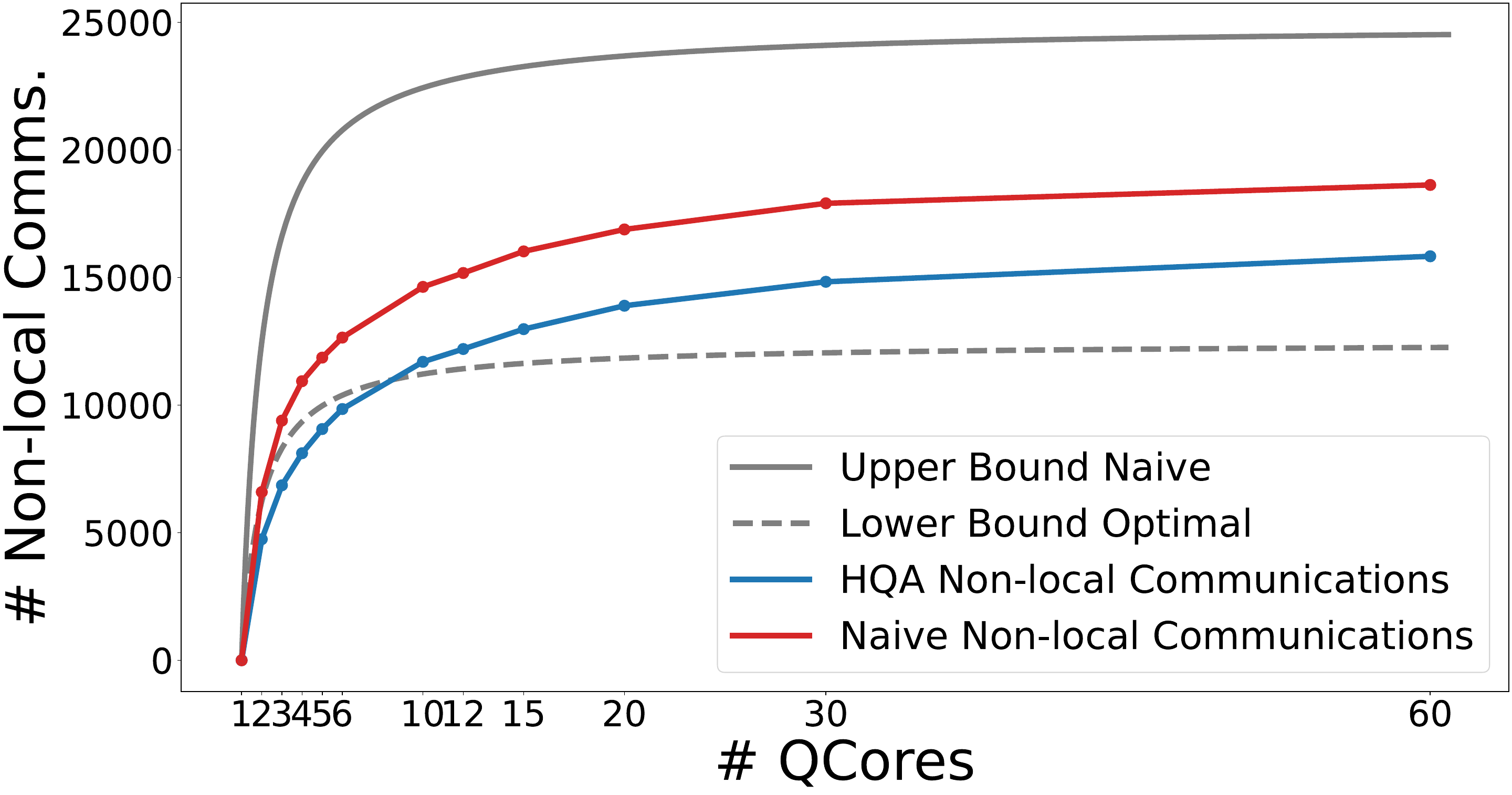}
    \caption{Random L HQA Mapping.}
    \label{fig:random_l_hqa}
\end{subfigure}
\caption{Non-local communications when using the HQA Algorithm for different Quantum Random Circuits. Communication bounds are shown in grey, and Naive Mapping is shown in red.}
\label{fig:hqa_bound_experiments}
\end{figure}

We attribute such improvement to the change of approach the proposed algorithm uses. Instead of partitioning a graph that describes the quantum circuit, the HQA assigns unfeasible two-qubit gates to cores, which directly focuses on minimizing the number of non-local communications.

Figure \ref{fig:attraction} depicts a comparison of the two proposed versions of the HQA algorithm, with and without considering future qubit interactions for the cost function. It can be seen that, by considering future qubit interaction and computing the attraction force as formulated in Equation (\ref{cost function with attraction}), we manage to decrease the number of non-local communications for a structured (Cuccaro) and an unstructured (Random) quantum circuits.

Moreover, we also study the impact of the initial distribution of qubits to cores for the algorithm. Figure \ref{fig:initial_placement_multi} shows the number of non-local communications when using a random initial partition or the graph partition obtained using the OEE \cite{PARK1995899}, as proposed in \cite{baker_time-sliced_2020}. We show that the mapping algorithm performs worse when starting with a random partition than when starting with a partition obtained by the OEE algorithm.

\begin{figure}
\centering
\begin{subfigure}[t]{0.495\textwidth}
    \includegraphics[width=\textwidth]{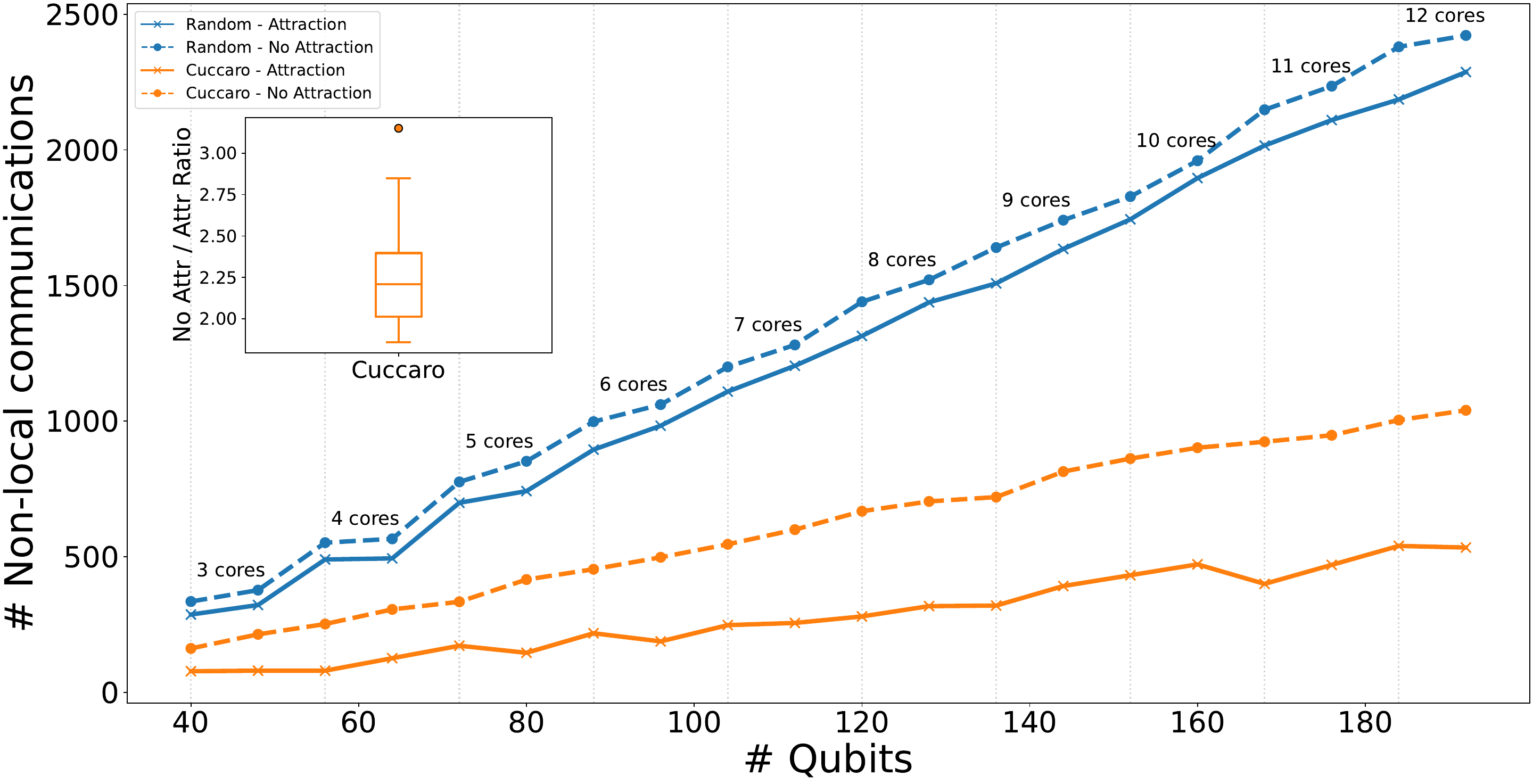}
    \caption{HQA communications with and without using the attraction force for future qubit interactions. \textcolor{black}{The use of attraction forces improves the number of non-local communications by $2.27\times$ for the Cuccaro Adder and by $1.09\times$ for the Random Circuit, on average.}}
    \label{fig:attraction}
\end{subfigure}
\hfill
\begin{subfigure}[t]{0.495\textwidth}
    \includegraphics[width=\textwidth]{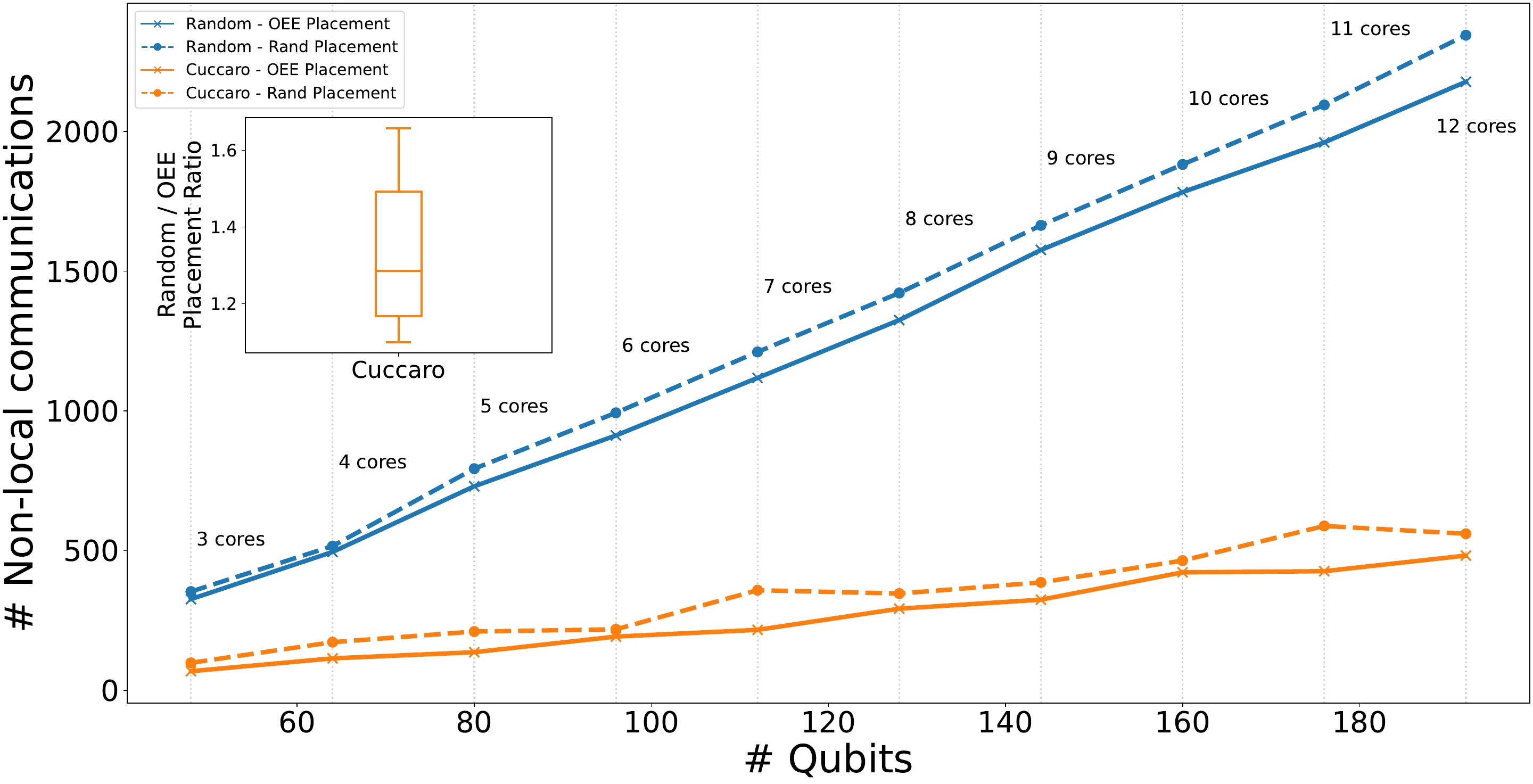}
    \caption{HQA communications with Random initial placement and OEE initial placement. \textcolor{black}{The use of the OEE initial placement improves the number of non-local communications by $1.33\times$ for the Cuccaro Adder and by $1.07\times$ for the Random Circuit, on average.}}
    \label{fig:initial_placement_multi}
\end{subfigure}
\caption{Different versions of the HQA, mapped into an architecture with 16 qubits per core and as many cores as needed, depending on the number of qubits used ($x$-axis). A structured circuit (Cuccaro Adder) and an unstructured one (Random) have been selected to highlight the importance of the attraction force. In each figure, the box plot depicts the Cuccaro Adder communications ratio with the two approaches.}
\label{fig:hqa_parameters}
\end{figure}

Therefore, from the results depicted in Figure \ref{fig:hqa_parameters}, we will use the HQA considering future qubit interactions and starting with a partition for the experiments carried out in future sections, where we deeply analyze the performances of the three reviewed mapping algorithms, assessing them in terms of optimally (number of non-local communications) and efficiency (execution time).

\section{Performance evaluation of HQA}
\label{sec:methodology}
To compare the different mapping algorithms explained before, we will focus on the scalability problem, increasing the number of cores and qubits, to analyze how the different approaches adapt to the increasing of resources. To do so, three different sets of experiments are proposed:

\begin{itemize}
    \item For the first scalability approach, \textbf{virtual scaling}, a fixed-size multi-core quantum computing architecture of 10 cores and 10 qubits per core (100 physical qubits in total) has been set, it consists of mapping a quantum circuit with increasing the number of virtual qubits from 50 virtual qubits, where half of the qubits of the architecture will be used as ancillary qubits, to 100 virtual qubits, where all qubits from the architecture will be used as data qubits and no ancillary qubits will be used. With this experiment, we aim to assess the importance of ancillary qubits in multi-core mapping.
    \item The \textbf{weak scaling} approach involves mapping a quantum circuit of 200 virtual qubits, into a quantum hardware with 200 physical qubits, varying the numbers of cores (from 2 to 10 cores) and qubits per core. The number of qubits per core will depend on the architecture’s number of cores, ensuring that all cores will have the same number of qubits. Therefore, no ancillary qubits will be used in any hardware configuration for this approach. This experiment demonstrates how the number of communications across cores and execution time varies with the number of architecture cores and how the selected mapping algorithms perform when increasing the number of architecture cores. Due to the high number of qubits used in this experiment, we only show the performance of the proposed Naive approach (Algorithm \ref{alg:naive_mapping}), the FGP-rOEE \cite{baker_time-sliced_2020}, and the proposed HQA (Algorithm \ref{alg:hqa_mapping}), excluding the QUBO Mapping algorithm \cite{bandic_mapping_2023} given its high runtime.
    \item Lastly, the last scaling approach, \textbf{strong scaling}, consists of increasing the number of cores (with 10 qubits each) in the architecture, from 2 cores (20 qubits) up to 20 cores (200 qubits), in each case the number of virtual qubits will be the same as the physical qubits. Again, due to the high execution time of the QUBO mapping algorithm \cite{bandic_mapping_2023}, the set of experiments using this algorithm will range only up to 10 cores (100 qubits). With this experiment, we show how the algorithms behave when increasing both the number of qubits and cores, covering all the different types of scaling.
\end{itemize}

\subsection{Benchmarks and performance metrics}
For all scaling approaches, the same set of benchmarks is selected. Containing high-structured quantum algorithms such as Quantum Fourier Transform \cite{coppersmith2002approximate, nielsen_chuang_2010}, Draper Adder \cite{draper2000addition}, and Cuccaro Adder \cite{cuccaro2004new}, as well as unstructured quantum algorithms such as Random Quantum Circuits (with a depth two times the number of virtual qubits used), and Quantum Volume \cite{cross_2019_validating}.

We have used Qiskit's \cite{Qiskit} implementation of the algorithms. Each algorithm is sliced into timeslices as proposed by \cite{bandic_mapping_2023}, and the same set of slices is used as input for each mapping algorithm.

For the HQA experiments, we use the cost function that considers future qubit interactions (described in Equation (\ref{cost function with attraction})), and the initial assignments of qubits into cores will be the one obtained by the OEE \cite{PARK1995899} applied to the total interaction graph. Both decisions are supported by the exploration conducted in Section \ref{sec:HQA_discussion}, and summarized in Figure \ref{fig:hqa_parameters}.

The main performance metric used in this work is the number of non-local communications (i.e. communications across cores). The Naive Algorithm proposed in Algorithm \ref{alg:naive_mapping} will be used as a baseline for the number of non-local communications.

Moreover, we also analyze the execution time of each mapping algorithm, showing the speedup of the HQA Algorithm compared to FGP-rOEE and QUBO Mapping algorithms.

\textcolor{black}{All experimental procedures were conducted on a computing system featuring an Intel(R) Xeon(R) CPU E5-2640 v4 @ 2.40GHz, equipped with 131.7 GB of RAM and 40 cores, operating on CentOS Linux 7. The simulation procedure was implemented utilizing Python 3.8, and the benchmarking tasks were facilitated through the framework provided by Qiskit \cite{Qiskit}. The QUBO implementation was obtained from its open repository \cite{bandic_mapping_2023}, and FGP-rOEE has been implemented according to its proposal in \cite{baker_time-sliced_2020}.}

\subsection{Results}
\label{sec:results}
This section analyses the three different mapping algorithms based on the different scaling experiments explained in the last section: Virtual, Weak, and Strong Scaling. For each approach, we show the number of non-local communications for all the selected benchmarks, as well as an improvement ratio of FGP-rOEE, QUBO or HQA over the Naive approach (median across the different data points).

Moreover, the execution time of all mapping algorithms for each benchmark is depicted in the lower row of Figures \ref{fig:virtual_scaling}, \ref{fig:weak_scaling}, and \ref{fig:strong_scaling}, along with the average execution time, over all the benchmarks. Lastly, a SpeedUp of the HQA over FGP-rOEE and QUBO is shown, highlighting the HQA algorithm's value.
\subsubsection{Virtual Scaling}
Figure \ref{fig:virtual_scaling} shows, in the first two rows, the non-local communications needed for each benchmark when increasing the number of virtual qubits of the circuit ($x$-axis). Results show how the number of non-local communications increases when increasing the number of virtual qubits in the circuit, independently of the quantum benchmark or the mapping algorithm. This is due to the use of ancillary qubits (physical qubits of the architecture that do not store any quantum state) for allocation purposes.

The HQA algorithm outperforms all other mapping algorithms in all benchmarks except for the Cuccaro Adder, where a similar number of non-local communications is obtained for HQA, FGP-rOEE and the Naive approach. For most benchmarks, FGP-rOEE obtains a similar number of non-local communications than the Naive approach, highlighting that the graph partitioning approach does not perform as well as the two-qubit operations assignment approach.

Regarding the execution time, we can see that HQA outperforms QUBO in all benchmarks, and FGP-rOEE in most of them, as, for Random Circuits, the compilation time of HQA is $\sim 0.67 \times$ faster ($\sim 1.5 \times$ slower) than FGP-rOEE. For all the other benchmarks, HQA managed to obtain the assignment with a speedup of up to $15 \times$. The high execution time for the QUBO Mapping algorithm was expected, as the graph-partition problem is posed as a quadratic optimization problem. 

\begin{figure}[t]
\includegraphics[width=\textwidth]{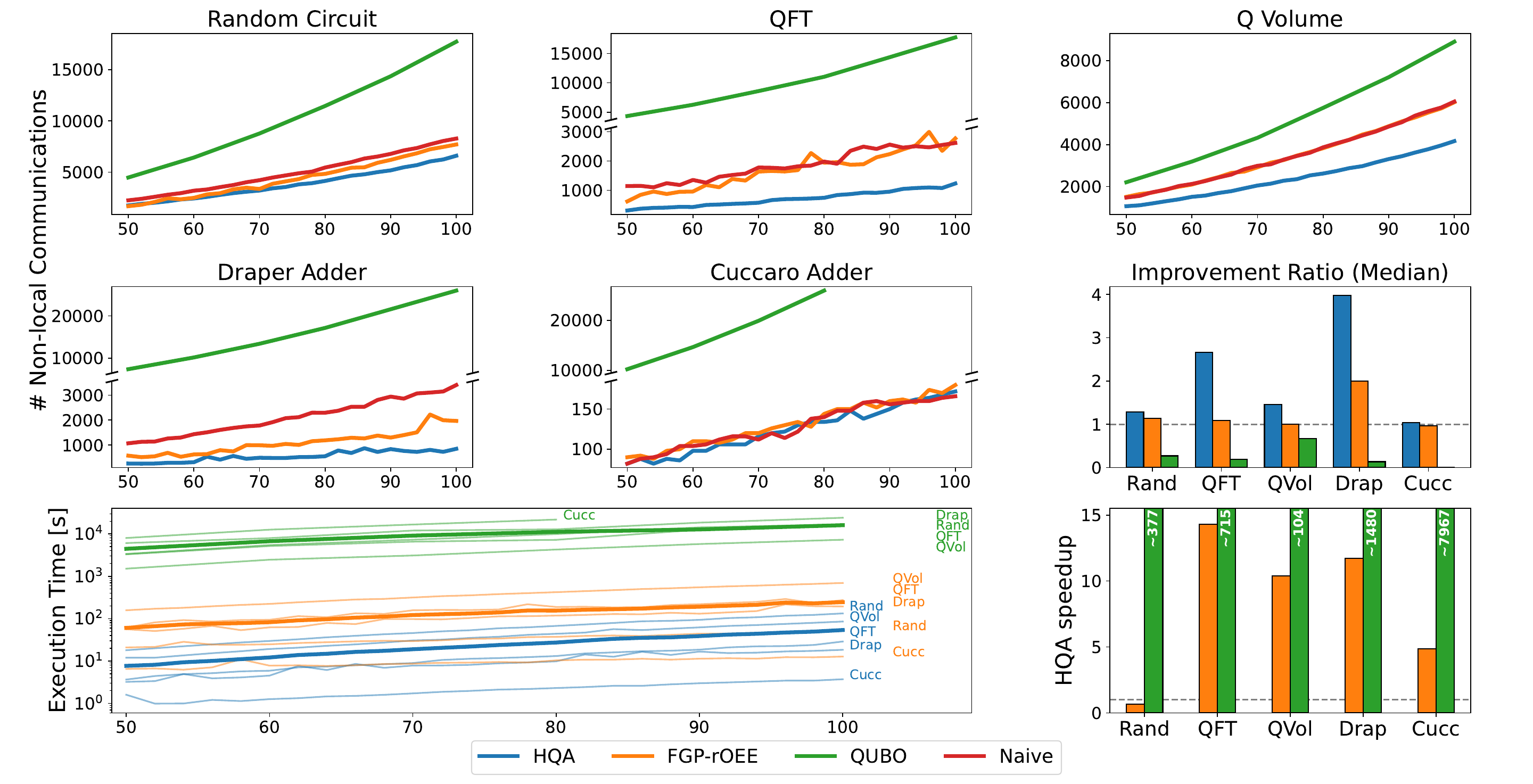}
\caption{Virtual Scaling. Mapping quantum algorithms into a 10-core architecture with ten qubits per core (100 total qubits), increasing the number of qubits used in the circuit from 50 qubits to 100 qubits. \textcolor{black}{The first two rows show the number of non-local communications for each one of the approaches, and the last row show the execution time for FGP-rOEE, QUBO, and HQA. The bar plots show the improvement ratio for both the number of non-local communications and the execution time. }}
\label{fig:virtual_scaling}
\end{figure}

\subsubsection{Weak Scaling}
Figure \ref{fig:weak_scaling} shows, in the first two rows, the non-local communications needed for each benchmark when increasing the number of cores of the architecture ($x$-axis), and the execution time of each experiment in the bottom row. It can be seen how, for all benchmarks, for all mapping algorithms, the number of non-local communications increases as so it does the number of cores of the architecture. Due to the high amount of qubits used in this experiment (architecture with two hundred physical qubits), the QUBO algorithm has been discarded, as it took too much time to execute.

With this scaling approach, we aim to study how mapping algorithms perform under edge cases, such as having an architecture with a hundred cores, and two qubits per core. Under these cases, we can see that, independently of the benchmark, FGP-rOEE performs much worse than the HQA algorithm and the Naive baseline. This is due again to the graph-partitioning approach this algorithm uses, as it performs swaps of qubits among cores, until obtaining a valid assignment. When the number of cores is too high, the possible swaps of qubits increase, making the heuristics fail in choosing which pair of qubits to swap.

Thanks to the new approach of assigning unfeasible two-qubit gates instead of partitioning the graph, the HQA algorithm achieves a lower number of non-local communications than the FGP-rOEE in most benchmarks, improving the Naive baseline in all of them.

Regarding the execution time, we can see that the HQA found the solution in less or equal time than FGP-rOEE for all benchmarks, remarking a speedup of more than $8 \times$ for the Cuccaro Adder benchmark. 

\begin{figure}[t]
\includegraphics[width=\textwidth]{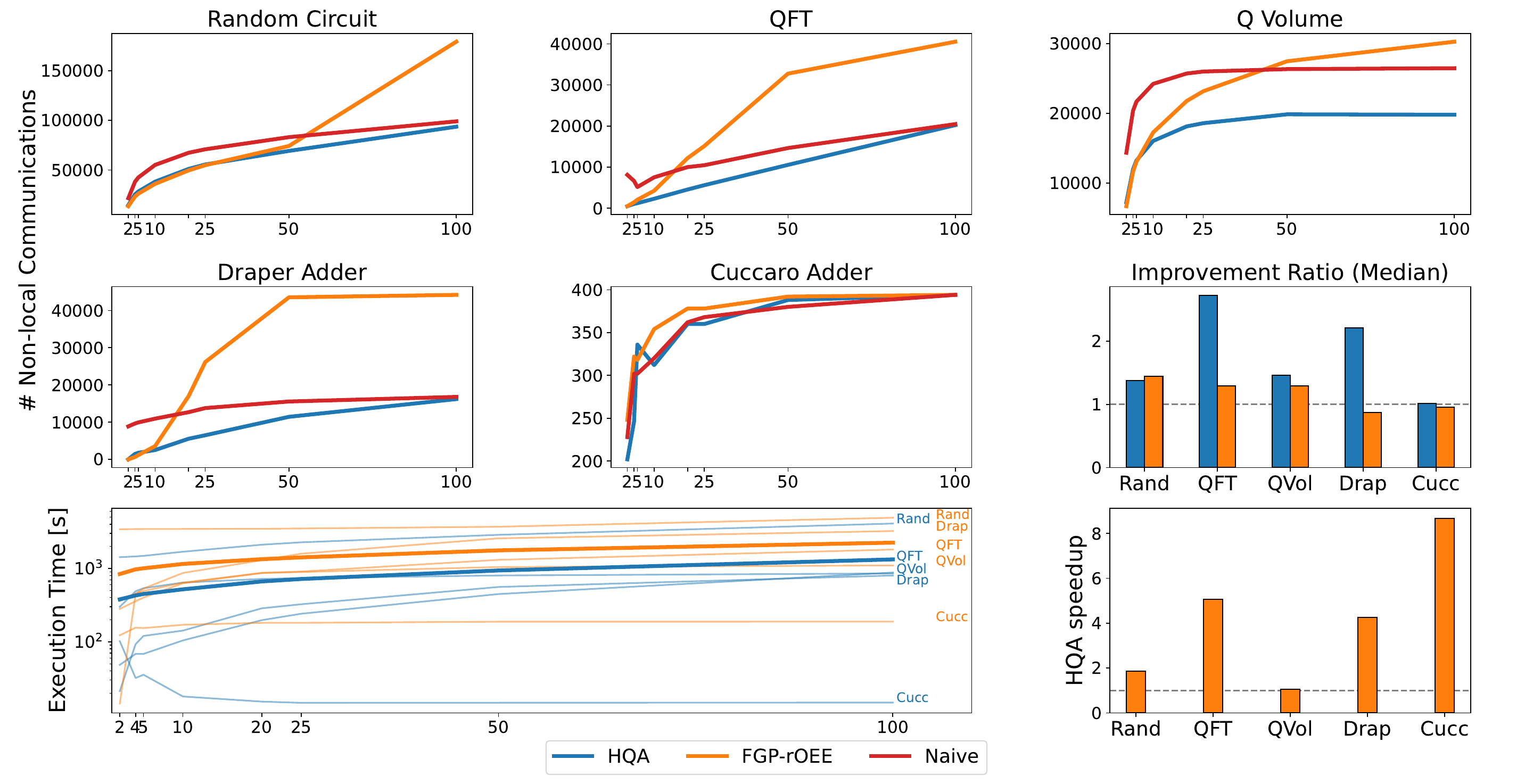}
\caption{Weak Scaling. Mapping quantum algorithms into a 200-qubit architecture with varying numbers of cores. \textcolor{black}{The first two rows show the number of non-local communications for each one of the approaches, and the last row show the execution time for FGP-rOEE and HQA. The bar plots show the improvement ratio for both the number of non-local communications and the execution time. }}
\label{fig:weak_scaling}
\end{figure}

\subsubsection{Strong Scaling}
Lastly, Figure \ref{fig:strong_scaling} shows, in the first two rows, the non-local communications needed for each benchmark when increasing the number of cores and qubits of the architecture ($x$-axis) and the execution time of each experiment in the bottom row. In this experiment, the size of each core is fixed to ten qubits per core, increasing the number of qubits in the circuit according to the number of qubits in the architecture. Due to the high execution cost of the QUBO algorithm, we have restricted the experiments from a 2-core architecture with 20 qubits up to an 8-core architecture with 80 qubits, while for the FGP-rOEE and the HQA mapping algorithms, we scale the architecture up to a 20-core architecture with 200 qubits.

As expected, it can be seen how, as we increase the amount of qubits in the circuit (and in the architecture), the number of non-local communications also increases. For all benchmarks, the HQA manages to obtain a lower number of non-local communications than the baseline, except for the Cuccaro Adder, where it achieves similar non-local communications.

The HQA also outperforms FGP-rOEE and QUBO in all benchmarks. The QUBO algorithm always performs worse than the baseline, and the FGP-rOEE algorithm performs similarly to the baseline for most benchmarks. On average, across all benchmarks and scaling approaches, the HQA improves the number of non-local communications obtained by FGP-rOEE by $1.556\times$.

Regarding the execution time, again, we see a high execution time for the QUBO algorithm, as it poses a graph-partitioning problem using quadratic optimization. On average, for a low number of cores, the FGP-rOEE mapping algorithm obtains a faster solution than the HQA. However, as the number of cores increases, the FGP-rOEE execution time increases much faster than the HQA's, \textcolor{black}{as it mostly depends on the number of cores and the number of qubits of the system,} ending with a lower execution time for the HQA mapping algorithm, highlighting its scalability. \textcolor{black}{For Random Circuits and Quantum Volume, which are unstructured circuits, the number of unfeasible two-qubit gates drive the HQA’s computation time, leading to a higher execution time. Nevertheless, even for unstructured circuits, FGP-rOEE takes more time than HQA once the number of cores or qubits is high enough.}

Without taking into account the HQA, FGP-rOEE is the best-performing multi-core mapping algorithm. The Hungarian Qubit Assignment algorithm managed to improve FGP-rOEE's number of non-local communications by $1.556\times$, reducing its execution time by $4.897\times$, on average, across all benchmarks and scalability approaches.

\begin{figure}[t]
\includegraphics[width=\textwidth]{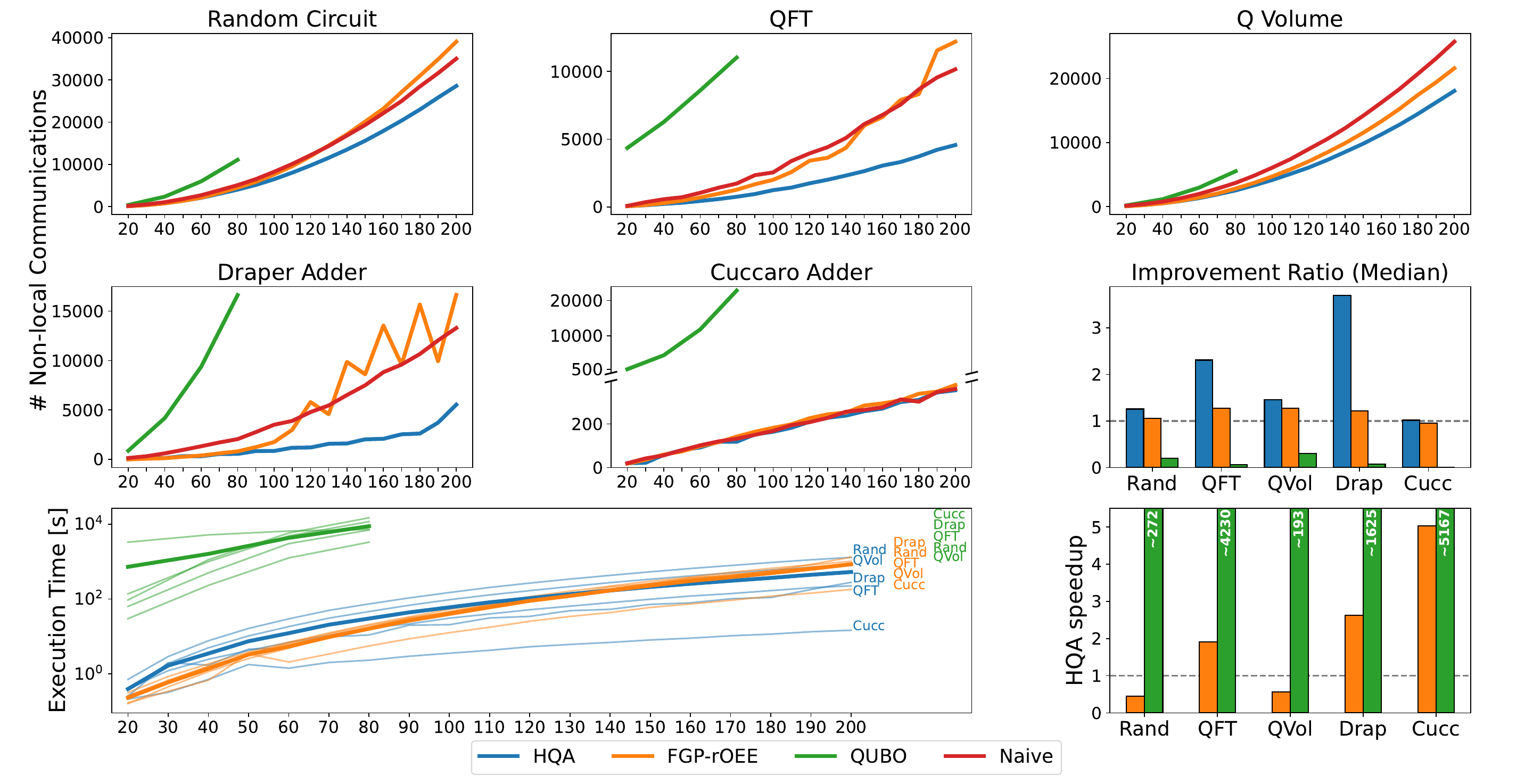}
\caption{Strong Scaling. Mapping quantum circuits into a $k$-core architecture, increasing the number of cores, with a fixed size of 10 qubits per core. \textcolor{black}{The first two rows show the number of non-local communications for each one of the approaches, and the last row show the execution time for FGP-rOEE, QUBO, and HQA. The bar plots show the improvement ratio for both the number of non-local communications and the execution time. }}
\label{fig:strong_scaling}
\end{figure}

\textcolor{black}{Figure \ref{fig:hqa_scaling} shows the behaviour of HQA when scaling up to 1024 qubits, also in a strong scaling way, with 64 qubits per core and an increasing number of cores.}

\begin{figure}
\centering
\begin{subfigure}[t]{0.48\textwidth}
    \includegraphics[width=\textwidth]{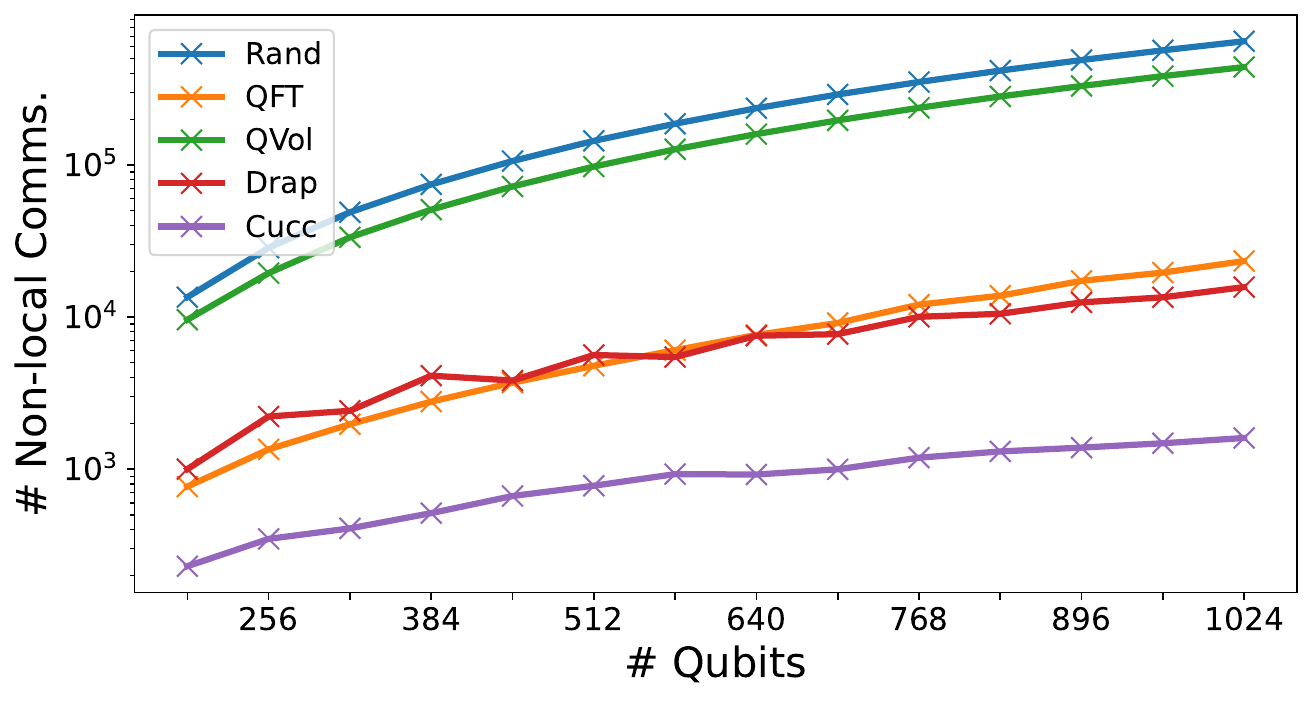}
    \caption{\textcolor{black}{HQA number of communications.}}
    \label{fig:hqa_scaling_comms}
\end{subfigure}
\hfill
\begin{subfigure}[t]{0.48\textwidth}
    \includegraphics[width=\textwidth]{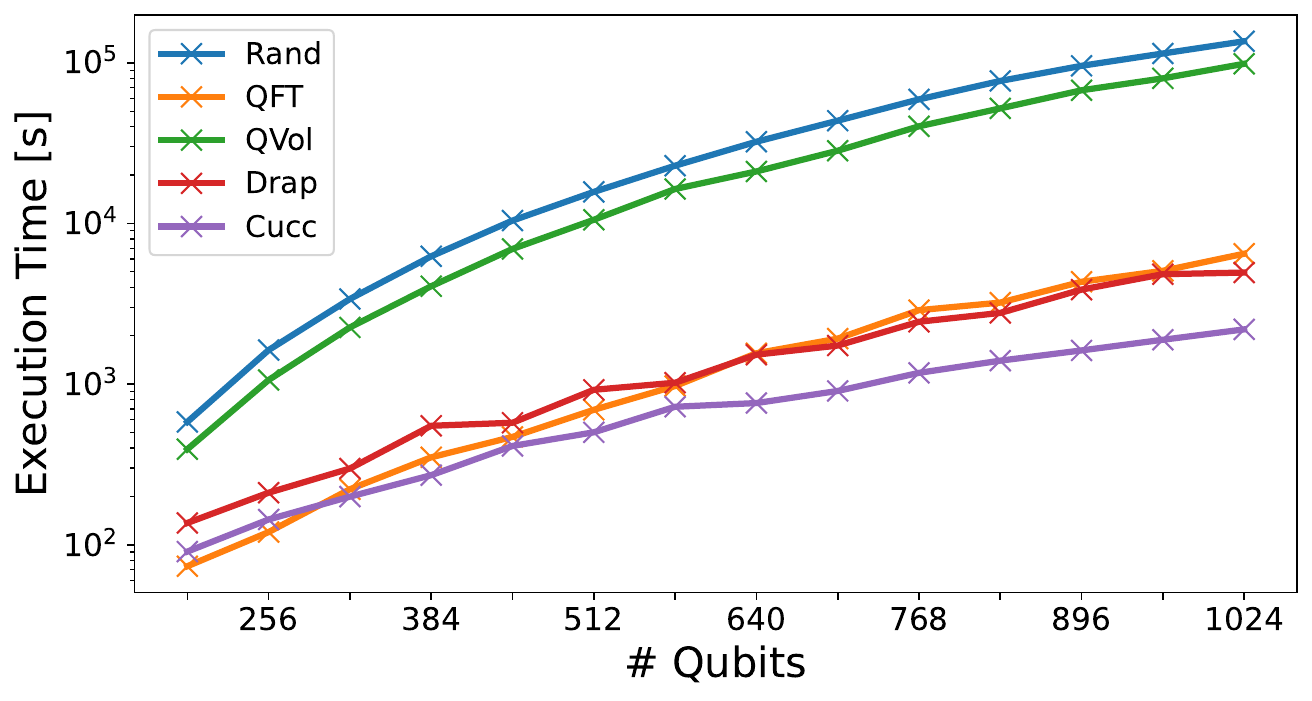}
    \caption{\textcolor{black}{HQA execution time.}}
    \label{fig:hqa_scaling_times}
\end{subfigure}
\caption{\textcolor{black}{HQA number of communications and execution time for a system with a fixed core size of 64 qubits.}}
\label{fig:hqa_scaling}
\end{figure}

\section{Conclusions and Future Work}
\label{sec:conclusions}
Multi-core quantum computing architectures offer a promising path to overcoming the limitations of monolithic processors and enabling the execution of complex quantum algorithms. In this paper, we have explored the intricate landscape of multi-core quantum computing, illustrating the challenges and opportunities ahead.

In this work, we have proposed theoretical bounds on the non-local communications needed to map Random Circuits into modular architectures, proving state-of-the-art multi-core mappers to be far from optimal. Moreover, we propose a novel approach, changing the paradigm of multi-core mapping from graph partitioning to two-qubit gate assignment. Throughout rigorous evaluation across different quantum algorithms and scaling approaches, we have shown the potential of HQA, obtaining better results than its analogous mapping algorithms ($1.6\times$ improvement over FGP-rOEE) while decreasing the execution time ($4.9\times$ speedup over FGP-rOEE), showcasing its potential to be scaled up.

Much more needs to be done in mapping for multi-core quantum computers. Our approach focuses on those modular architectures based on the generation and distribution of entangled states. However, other modular approaches will require different mapping algorithms that use other communication primitives than the ones used in this work. Moreover, other EPR-based communication primitives can be added to the mapping algorithm to take advantage of all types of quantum communications. 

\textcolor{black}{The main reason for scalability in quantum systems is to incorporate quantum error correction \cite{nielsen_chuang_2010} (QEC) and fault-tolerant techniques that will allow reliable and accurate computations. Therefore, it is crucial to investigate not only how to integrate QEC in these new modular architectures but also what runtime and compiler support will be required considering the constraints imposed by various correction codes, leveraging them to enhance the mapping optimization process.}

\section*{Acknowledgments}

Authors gratefully acknowledge funding from the European Commission through HORIZON-EIC-2022-PATHFINDEROPEN-01-101099697 (QUADRATURE) and grant HORIZON-ERC-2021-101042080 (WINC). This work has been partially supported by the Spanish Ministerio de Ciencia e Innovaci\'{o}n and European ERDF under grant PID2021-123627OB-C5.

\bibliographystyle{ACM-Reference-Format}
\bibliography{bibliography}

\end{document}